\documentclass[11pt,a4paper]{article}

\usepackage[T1]{fontenc}
\usepackage[utf8]{inputenc}
\usepackage{lmodern}
\usepackage{amsmath,amssymb,amsthm,mathtools}
\usepackage{float}
\usepackage{booktabs,tabularx,array}
\usepackage{needspace}
\usepackage{enumitem}
\usepackage[margin=1in]{geometry}
\usepackage{microtype}
\usepackage{tikz}
\usetikzlibrary{arrows.meta,positioning}
\usepackage{hyperref}
\usepackage{authblk}
\usepackage[nameinlink,capitalize]{cleveref}
\hypersetup{hidelinks,pdftitle={Non-Adaptive Learning of Sparse Erd\H{o}s--R\'enyi Graphs via Affine Splitting}}

\newtheorem{theorem}{Theorem}
\newtheorem{lemma}{Lemma}
\newtheorem{proposition}{Proposition}
\newtheorem{corollary}{Corollary}
\newtheorem{fact}{Fact}
\theoremstyle{definition}
\newtheorem{definition}{Definition}
\newtheorem{example}{Example}
\crefname{fact}{Fact}{Facts}

\floatstyle{ruled}
\newfloat{algorithm}{htbp}{loa}
\floatname{algorithm}{Algorithm}
\crefname{algorithm}{Algorithm}{Algorithms}
\Crefname{algorithm}{Algorithm}{Algorithms}
\newcounter{algline}
\newlength{\algindent}
\newenvironment{algorithmic}[1][1]{%
  \small\setcounter{algline}{0}\setlength{\algindent}{0pt}%
  \begin{list}{}{%
    \setlength{\leftmargin}{2.2em}%
    \setlength{\labelwidth}{1.7em}%
    \setlength{\labelsep}{0.5em}%
    \setlength{\itemsep}{1pt}%
    \setlength{\parsep}{0pt}%
    \setlength{\topsep}{4pt}%
  }%
}{\end{list}}
\newcommand{\STATE}{%
  \stepcounter{algline}\item[\footnotesize\arabic{algline}:]%
  \hangindent=\algindent\hangafter=1\relax\hspace*{\algindent}}
\newcommand{\REQUIRE}{\item[]\textbf{Input:}\ }
\newcommand{\FOR}[1]{\STATE\textbf{for} #1 \textbf{do}\addtolength{\algindent}{1.2em}}
\newcommand{\ENDFOR}{\addtolength{\algindent}{-1.2em}\STATE\textbf{end for}}
\newcommand{\IF}[1]{\STATE\textbf{if} #1 \textbf{then}\addtolength{\algindent}{1.2em}}
\newcommand{\ELSE}{\addtolength{\algindent}{-1.2em}\STATE\textbf{else}\addtolength{\algindent}{1.2em}}
\newcommand{\ENDIF}{\addtolength{\algindent}{-1.2em}\STATE\textbf{end if}}
\newcommand{\RETURN}{\STATE\textbf{return}\ }

\newcommand{\F}{\mathbb F}
\newcommand{\Pp}{\mathbb P}
\newcommand{\E}{\operatorname{E}}
\newcommand{\Var}{\operatorname{Var}}
\newcommand{\Cov}{\operatorname{Cov}}
\newcommand{\ind}{\mathbf 1}
\newcommand{\PD}{\mathcal{PD}}
\newcommand{\kbar}{\bar k}
\newcommand{\Ehat}{\widehat E}
\newcommand{\cE}{\mathcal E}
\newcommand{\cG}{\mathcal G}
\newcommand{\cN}{\mathcal N}
\newcommand{\cT}{\mathcal T}
\newcommand{\cC}{\mathcal C}
\newcommand{\bigO}{\mathcal O}
\newcommand{\ER}{\mathrm{ER}}

\title{\bfseries Non-Adaptive Learning of Sparse Erd\H{o}s--R\'enyi Graphs via Affine Splitting}

\author[]{Hoang Ta}
\affil[]{\small Department of Computer Science, Hanoi University of Science and Technology, Vietnam}

\begin{document}
\maketitle

\begin{abstract}
Graph learning from edge-detecting queries concerns the reconstruction of an unknown edge set on a known vertex set. Each query reports whether a specified vertex subset contains at least one edge. We study non-adaptive schemes, in which all queries are fixed before any outcomes are observed, with the goal of achieving exact recovery using few queries and fast decoding. For general graphs on $n$ vertices with at most $k$ edges, non-adaptive recovery requires $\Omega(\min\{k^2\log n,n^2\})$ queries in the worst case, even when a small error probability is allowed. In this paper, we consider Erd\H{o}s--R\'enyi ($\ER$) graphs $G\sim\ER(n,q)$, with expected edge count $\kbar=q\binom{n}{2}$. Our scheme uses $\bigO(\kbar\log n)$ queries and achieves exact recovery in $\bigO(\kbar\log n)$ decoding time with probability tending to one throughout the regime $\kbar\to\infty$ and $\kbar=o(n^2)$. This improves the previous $\bigO(\kbar^{1+\delta}\log n)$ decoding guarantee for any fixed $\delta>0$, while maintaining the same query order. The guarantee also extends beyond the previously studied regime $\kbar=\Theta(n^{2\theta})$ with fixed $\theta\in(0,1)$. Our approach builds on the binary splitting method used in prior work, which organizes vertices into a hierarchy of successively smaller groups. We introduce three main changes: (i) we use random affine hash functions over a finite field to process each candidate pair in constant time; (ii) we apply the splitting procedure directly to the full graph, avoiding the need to combine solutions to multiple smaller graph-learning subproblems; and (iii) we bound the total decoding workload directly rather than deriving separate high-probability bounds on candidate counts at each level.
\end{abstract}

\section{Introduction}
\label{sec:introduction}

Recovering an unknown graph from indirect observations is a common problem in learning theory and combinatorial inference. We study \emph{edge-detecting queries}: each query selects a subset of vertices and reveals whether the subgraph induced by this subset contains at least one edge. Such queries arise, for example, in identifying which pairs of chemicals react, using experiments that indicate only whether some reaction occurs in a selected mixture~\cite{bouvel2005combinatorial}. The problem can also be viewed as a constrained form of group testing in which the items are the potential edges and each test contains all potential edges within a selected vertex subset~\cite{AJS26}. Our goal is to recover the edge set exactly using few queries and an efficient decoder.

Two query settings are commonly considered: \emph{adaptive} and \emph{non-adaptive}. In the adaptive setting, each query may depend on previous outcomes, allowing the learner to refine its search as information becomes available. The query complexity in this setting is well understood: a graph with $k$ edges contained in a known host graph $H$ can be recovered using $k\log_2(|E(H)|/k)+\bigO(k)$ queries, which is optimal up to an additive $\bigO(k)$ term~\cite{J02}. In particular, $\bigO(k\log n)$ queries suffice for arbitrary graphs on $n$ vertices. Randomized algorithms attain the same order of expected query complexity with a constant number of rounds when $k$ is known~\cite{AC08}, and with $\bigO(\log^* n)$ rounds when $k$ is unknown~\cite{AB19}. In the non-adaptive setting, all query subsets are fixed before any outcomes are observed, so the queries can be evaluated in parallel. This restriction substantially increases the query complexity in the worst case: recovering graphs with at most $k$ edges requires $\Omega(\min \{k^2\log n,n^2\})$ non-adaptive queries, even when a small probability of error is allowed~\cite{AB19,LFS19}.

These worst-case lower bounds motivate the study of recovery under a random graph model. We consider Erd\H{o}s--R\'enyi graphs $G\sim \ER(n,q)$, in which each potential edge is present independently with probability $q$, and we denote by $\kbar:=q\binom{n}{2}$ the expected number of edges. For $\kbar=\Theta(n^{2\theta})$ with fixed $\theta\in(0,1)$, the COMP and DD schemes of~\cite{LFS19} achieve vanishing error probability with $\bigO(\kbar\log n)$ queries, but their decoders require time at least quadratic in $n$. The same work establishes the converse bound $\Omega(\kbar\log(n^2/\kbar))$ for any non-adaptive design with vanishing error probability, so this number of queries is optimal up to constant factors in these regimes. It also proposes a scheme based on the GROTESQUE group testing algorithm~\cite{CJBJ17}, which reduces the decoding time to $\bigO(\kbar\log^2\kbar+\kbar\log n)$ at the cost of increasing the number of queries to $\bigO(\kbar\log\kbar\cdot\log^2 n)$~\cite{LFS19}. Recently, non-adaptive binary splitting~\cite{PS20,CN20} was adapted to this problem in~\cite{TS25}, achieving $\bigO(\kbar\log n)$ queries and $\bigO(\kbar^{1+\delta}\log n)$ decoding time for any fixed $\delta>0$, under the same polynomial sparsity assumption. Whether the decoding time can be reduced to $\bigO(\kbar\log n)$ was left open there, and the analysis does not cover slowly growing edge counts such as $\kbar=\bigO(\log n)$~\cite{TS25}.

In this work, we resolve this question by achieving $\bigO(\kbar\log n)$ queries and $\bigO(\kbar\log n)$ decoding time simultaneously, for every sequence satisfying $\kbar\to\infty$ and $\kbar=o(n^2)$. Both exact recovery and the decoding-time bound hold with probability tending to one.

\paragraph{Challenges.}
The binary splitting framework partitions the vertices into groups, called \emph{blocks}, and tracks candidate pairs of blocks as the partition is refined down to single vertices. Achieving $\bigO(\kbar\log n)$ decoding time requires addressing two challenges. First, processing a candidate pair in the basic design of~\cite{TS25} may require scanning $\Theta(\sqrt{\kbar})$ query outcomes. We introduce an affine test design that allows each candidate pair to be processed in constant time using at most one outcome lookup. Second, the total number of candidates must be controlled despite dependencies among their survival events. Edges within blocks create an additional difficulty: every query containing such a block is positive, so many candidate pairs may survive. We address these effects through covariance bounds for candidate survival and a bound on the cumulative contribution of internal edges across levels. These estimates allow us to bound the total decoding workload directly, instead of combining separate high-probability bounds at each level. The resulting analysis applies throughout the regime $\kbar\to\infty$ and $\kbar=o(n^2)$, including slowly growing expected edge counts.

\subsection{Main Result and Contributions}
\label{sec:contributions}

Our main result, stated formally in~\cref{thm:main}, is the following.
\begin{quote}
\textbf{Main result.}
Let $G\sim\ER(n,q)$, where $\kbar=q\binom{n}{2}$ satisfies $\kbar\to\infty$ and $\kbar=o(n^2)$. There exists a non-adaptive scheme using $\bigO(\kbar\log n)$ edge-detecting queries that, with probability $1-o(1)$ over the graph and the random design, recovers the edge set exactly in $\bigO(\kbar\log n)$ decoding time.
\end{quote}
The decoding bound holds under the computational assumptions in~\cref{sec:problem}, starting from the stored design and query outcomes. Our contributions are as follows.
\begin{enumerate}[label=(\roman*),itemsep=3pt]
\item \emph{Faster decoding with the same query order.}
We improve the $\bigO(\kbar^{1+\delta}\log n)$ decoding guarantee of~\cite{TS25}, for any fixed $\delta>0$, to $\bigO(\kbar\log n)$ while retaining $\bigO(\kbar\log n)$ queries.

\item \emph{A broader sparsity regime.}
The recovery guarantees of~\cite{LFS19,TS25} compared here assume $\kbar=\Theta(n^{2\theta})$ with fixed $\theta\in(0,1)$. Our guarantees require only $\kbar\to\infty$ and $\kbar=o(n^2)$. 

\item \emph{An affine design and a global workload analysis.}
We retain the block hierarchy and refinement rule of~\cite{TS25}, with new query assignments defined by random affine lines over a finite field. Within each family, two blocks with distinct slopes have a unique shared query that can be located in constant time; pairs with equal slopes are retained automatically. We combine covariance estimates for candidate survival with a bound on the cumulative contribution of internal edges across levels. These estimates control the total decoding workload with high probability throughout the stated sparsity regime.
\end{enumerate}

The query count matches the information-theoretic lower bound $\Omega(\kbar\log(n^2/\kbar))$ of~\cite{LFS19} up to constant factors whenever $\log(n^2/\kbar)=\Omega(\log n)$. \Cref{tab:comparison} compares our guarantees with those of previous schemes.

\begin{table}[!htbp]
\centering
\small
\setlength{\tabcolsep}{4pt}
\renewcommand{\arraystretch}{1.3}
\begin{tabularx}{\textwidth}{|>{\raggedright\arraybackslash}X|>{\centering\arraybackslash}p{3.0cm}|>{\centering\arraybackslash}p{3.8cm}|>{\centering\arraybackslash}p{2.8cm}|}
\hline
\textbf{Reference} & \textbf{Queries} & \textbf{Decoding time} & \textbf{Sparsity regime} \\
\hline
COMP/DD~\cite{LFS19} & $\bigO(\kbar\log n)$ & $\bigO(n^2\kbar\log n)$ & $\kbar=\Theta(n^{2\theta})$ \\
\hline
GROTESQUE-based~\cite{LFS19} & $\bigO(\kbar\log\kbar\log^2 n)$ & $\bigO(\kbar\log^2\kbar+\kbar\log n)$ & $\kbar=\Theta(n^{2\theta})$ \\
\hline
Binary splitting~\cite{TS25} & $\bigO(\kbar\log n)$ & $\bigO(\kbar^{1+\delta}\log n)$ & $\kbar=\Theta(n^{2\theta})$ \\
\hline
\textbf{This work} & $\bigO(\kbar\log n)$ & $\bigO(\kbar\log n)$ & $\kbar\to\infty$, $\kbar=o(n^2)$ \\
\hline
\end{tabularx}
\caption{Comparison of non-adaptive recovery guarantees with vanishing error probability. The parameters $\theta\in(0,1)$ and $\delta>0$ are fixed.}
\label{tab:comparison}
\medskip
\end{table}

\subsection{Related Work}
\label{sec:related}

\emph{Learning graphs and hypergraphs with edge-detecting queries.}
Early work on edge-detecting queries considered searching for a single hidden edge~\cite{AT88} and identifying several defective edges within a known host graph~\cite{J02}. Subsequent work studied the learning of general graphs under fully adaptive queries, a bounded number of adaptive rounds, and non-adaptive designs~\cite{AC08,AB19}. The problem has also been extended to hidden hypergraphs~\cite{AC06,ABM18}. For random graphs, the non-adaptive schemes of~\cite{LFS19,TS25} discussed above circumvent the worst-case lower bound of~\cite{AB19}. In particular, the binary splitting framework of~\cite{TS25} organizes the vertices into a hierarchy of blocks and refines surviving block pairs from coarse to fine levels; our scheme builds on the same framework. This line of work has also been extended to random hypergraphs: non-adaptive learning of random $k$-uniform hypergraphs was studied in~\cite{ART25}, and a hierarchical splitting approach for random hypergraphs was developed in~\cite{PT26}.

\emph{Efficiently decodable group testing.}
Efficient decoding is also a central objective in standard group testing; see~\cite{AJS26} for a survey. GROTESQUE~\cite{CJBJ17}, SAFFRON~\cite{LCPR19}, and bit-mixing coding~\cite{BCSYZ21} offer different tradeoffs between the number of tests and the decoding time. For $K$ defective items among $N$ items, non-adaptive binary splitting achieves $\bigO(K\log N)$ tests and $\bigO(K\log N)$ decoding time in the small-error setting~\cite{PS20}, and related splitting constructions attain an optimal number of tests with near-optimal decoding time in the zero-error setting~\cite{CN20}. Subsequent work improved the constants in the number of tests~\cite{WGG24} and addressed sparsity constraints and noisy outcomes~\cite{PST23,LM25}. These schemes do not transfer directly to graph learning: when potential edges are viewed as items, an admissible test must contain all potential edges induced by some vertex subset, which rules out arbitrary test designs.

\emph{Algebraic test designs and workload analysis.}
Our construction is also related to algebraic test designs. The Kautz--Singleton construction~\cite{KS64} converts evaluations of polynomials over a finite field into binary test assignments, and its performance in probabilistic group testing was analyzed in~\cite{IKWO19}. Our affine incidence rule is of this form with degree-one polynomials, whose coefficients are sampled independently for each block. Here, however, the algebraic structure serves a different purpose: it locates the unique shared query of a pair of blocks with distinct slopes. Likewise, bounding the total number of candidates visited by a splitting decoder is an established technique in group testing~\cite{PS20,CN20,WGG24}. Our analysis extends it to account for the dependencies between block pairs and for the cumulative contribution of internal edges across all levels of the hierarchy.


\subsection{Technical Overview}
\label{sec:overview}

\emph{Typical graphs.}
We first define a \emph{typical set}, consisting of graphs satisfying the structural bounds needed for our analysis, and show that $G\sim\ER(n,q)$ belongs to this set with probability $1-o(1)$. These bounds refer to a fixed hierarchy of vertex groups, called \emph{blocks}, starting from $\Theta(\sqrt{\kbar})$ blocks and repeatedly splitting each block into two children until singletons remain. A typical graph has an edge count close to $\kbar$, a total of $\bigO(\sqrt{\kbar})$ edges inside the initial blocks, and $\bigO(\sqrt{\kbar})$ edges leaving any block at any level. Its internal-edge counts, weighted by the number of blocks at each level, also sum to $\bigO(\kbar\log n)$. We establish these properties using concentration estimates and then analyze the algorithm for an arbitrary fixed typical graph.

\emph{Affine query assignments.}
The design uses independent families of $p^2$ queries, where $p=\Theta(\sqrt{\kbar})$ is prime. Within each family, the queries are indexed by $(r,s)\in\F_p^2$. Each block $u$ receives an affine map $P_u(r)=a_ur+b_u$, with coefficients chosen independently and uniformly from $\F_p$, and belongs to query $(r,s)$ precisely when $P_u(r)=s$. If $a_u\ne a_v$, the two blocks share exactly one query within the family, whose coordinates are
\[
r_{uv}^*=(b_v-b_u)(a_u-a_v)^{-1},
\qquad
s_{uv}^*=a_ur_{uv}^*+b_u.
\]
These coordinates can be computed in constant time under the assumptions in~\cref{sec:problem}. This allows direct access to a shared query, removing the scan through $\Theta(\sqrt{\kbar})$ repetitions in the basic decoder of~\cite{TS25}. We use one family at each non-final level and $R=\log n$ independent families at the singleton level. All families are sampled before observing any outcomes, giving $\bigO(\kbar\log n)$ non-adaptive queries in total.

\emph{Decoding by binary refinement.}
The decoder follows the refinement rule of~\cite{TS25}, building on non-adaptive splitting for group testing~\cite{PS20,CN20}. It starts from all block pairs at the coarsest level. For a candidate pair with distinct slopes, it reads the shared query outcome and discards the pair if that outcome is negative. Pairs with equal slopes are retained automatically. At a non-final level, each retained pair generates the six pairs among its four child blocks, with duplicates removed (\cref{fig:refinement}). The four cross pairs preserve edges between the parent blocks, while the two sibling pairs preserve edges within them. Thus, every true edge remains in the union of some candidate pair at each level. At the singleton level, the decoder applies the same retention rule for the $R$ final rounds without generating children, and returns the surviving vertex pairs.

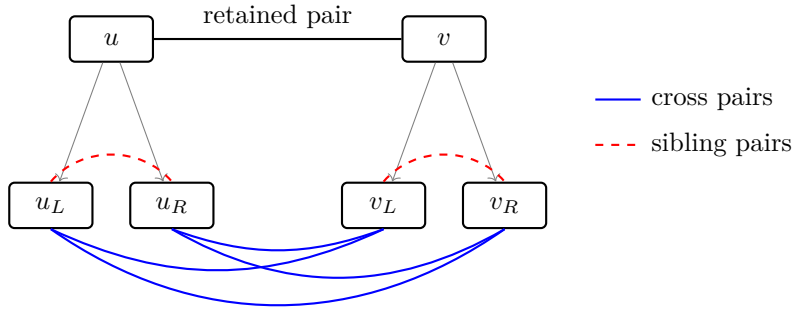
\begin{figure}[!htbp]
\centering
\begin{tikzpicture}[
  blk/.style={draw, thick, rounded corners=2pt, minimum width=1.1cm, minimum height=0.6cm},
  every node/.style={font=\small}]
  \node[blk] (U) at (0,2.2) {$u$};
  \node[blk] (V) at (4.4,2.2) {$v$};
  \draw[thick] (U) -- node[above] {retained pair} (V);
  \node[blk] (U1) at (-0.8,0) {$u_L$};
  \node[blk] (U2) at (0.8,0) {$u_R$};
  \node[blk] (V1) at (3.6,0) {$v_L$};
  \node[blk] (V2) at (5.2,0) {$v_R$};
  \foreach \c in {U1,U2} \draw[->, gray] (U) -- (\c);
  \foreach \c in {V1,V2} \draw[->, gray] (V) -- (\c);
  \draw[thick, blue] (U1.south) to[bend right=25] (V1.south);
  \draw[thick, blue] (U1.south) to[bend right=35] (V2.south);
  \draw[thick, blue] (U2.south) to[bend right=20] (V1.south);
  \draw[thick, blue] (U2.south) to[bend right=30] (V2.south);
  \draw[thick, red, dashed] (U1.north) to[bend left=50] (U2.north);
  \draw[thick, red, dashed] (V1.north) to[bend left=50] (V2.north);
  \draw[thick, blue] (6.4,1.4) -- (7.0,1.4) node[right, black] {cross pairs};
  \draw[thick, red, dashed] (6.4,0.8) -- (7.0,0.8) node[right, black] {sibling pairs};
\end{tikzpicture}
\caption{Refinement of a retained pair $\{u,v\}$ into the six unordered pairs among its four child blocks.}
\label{fig:refinement}
\end{figure}

\emph{Survival probabilities and dependencies.}
A pair is non-defective if the union of its blocks contains no edge. At each non-final level, a retained pair generates at most six child pairs. Since a non-defective pair survives with probability at most $1/24$, it contributes at most $1/4$ child pairs in expectation. To control dependencies, we bound the joint survival probability restricted to a separation event by the product of the marginal survival probabilities. Since separation fails with probability at most $7/p$, the covariance is bounded above by $7/p$, including for pairs sharing a block.

\emph{Bounding the total workload.}
The typical-graph bounds control the total contribution of defective pairs across all levels and final rounds by $\bigO(\kbar\log n)$. Combined with the conditional survival and covariance estimates, this yields first- and second-moment bounds for the candidate process. Centering the survivor counts by their means given earlier families produces errors with vanishing cross moments. A single deviation estimate for their accumulated sum then shows that the decoder examines $\bigO(\kbar\log n)$ candidates with high probability, giving the claimed decoding time. These bounds hold uniformly over typical graphs and therefore extend to the joint randomness of the graph and the design.

\paragraph{Organization.}
\Cref{sec:prelim} introduces the model, the computational assumptions, and the structural estimates for random graphs. \Cref{sec:design} presents the affine query design and the decoder. \Cref{sec:analysis} establishes the guarantees on the number of queries, exact recovery, and decoding time. All deferred and technical proofs are given in the appendix.

\section{Problem Setup and Preliminaries}
\label{sec:prelim}

In this section, we first introduce the notation used in the paper and then state the problem and the computational assumptions. Finally, we establish properties of the random graph used to design the testing and decoding algorithms and prove their performance guarantees.

\paragraph{Notation.}
Throughout, $\log$ denotes the base-$2$ logarithm, $\ln$ the natural logarithm, and $e$ Euler's number. For a finite set $X$, the collection of its unordered two-element subsets is denoted by $\binom X2$. The binomial distribution with parameters $m$ and $p'$ is written $\operatorname{Bin}(m,p')$, and the uniform distribution on $X$ is written $\operatorname{Unif}(X)$. Expectation is denoted by $\E[X]$, with square brackets enclosing its argument throughout; subscripts indicate the randomness being averaged whenever this is not clear from the context, and in particular $\E_t[X]$ is defined in~\cref{sec:global}. The relation $f\ll g$ means $f=o(g)$. Constants denoted by $c,C,C',\ldots$ are absolute and may change between occurrences unless explicitly fixed.. Finally, $\ind\{\cdot\}$ denotes the indicator of an event.


The unknown object is an undirected simple graph $G=(V,E)$ on the known vertex set $V=[n]:=\{1,\ldots,n\}$. Under the Erd\H{o}s--R\'enyi model $G\sim \ER(n,q)$, each potential edge in $\binom V2$ is present independently with probability $q=q(n)$. The expected number of edges is written
\[
\kbar:=q\binom n2 .
\]
The graph is sampled once and then held fixed; the learner has no direct access to its edge set.

\paragraph{Query model.}
An \emph{edge-detecting query} is specified by a set $S\subseteq V$ and returns
\[
Y(S):=\ind\!\left\{E\cap\binom S2\ne\varnothing\right\}.
\]
Equivalently, if $X\in\{0,1\}^n$ denotes the indicator vector of $S$, then
\[
Y(S)=\bigvee_{\{x,y\}\in E}(X_x\wedge X_y).
\]
Such queries are monotone: $Y(S)\le Y(S')$ whenever $S\subseteq S'$, and $Y(\varnothing)=0$.

\paragraph{Non-adaptive recovery.}
All query sets are chosen before any outcome is observed and are evaluated in a single batch. Given the design together with its outcomes, a decoder returns an estimate $\Ehat\subseteq\binom V2$, and its error probability is
\[
P_e:=\Pp[\Ehat\ne E],
\]
the probability being taken over both the random graph and the randomness of the design. The objective is exact recovery with $P_e\to0$, using few queries and few decoding operations.

\paragraph{Sparsity Regime.} Previous works~\cite{LFS19,TS25} consider the polynomial sparsity regime
\begin{equation}
q=\Theta\!\left(n^{-2(1-\theta)}\right),
\qquad
\kbar=\Theta(n^{2\theta}),
\qquad \text{ where }\theta\in(0,1)\text{ is fixed}.
\label{eq:sparsity}
\end{equation}
In this work, we consider the broader regime with
\begin{equation}
\kbar=o(n^2),
\text{ and }
\kbar\longrightarrow \infty \text{ as } n \to \infty.
\label{eq:regime}
\end{equation}
All our results are stated under~\eqref{eq:regime}, which includes slowly growing edge counts such as $\kbar=\Theta(\log n)$.

\subsection{Mathematical and Computational Assumptions}
\label{sec:problem}

Decoding time is measured in the unit-cost word-RAM model with words of $\Theta(\log n)$ bits. Vertex labels, block labels, test indices, and field elements each occupy a constant number of words, and basic arithmetic, comparisons, and array accesses take $\bigO(1)$ time. The decoding analysis rests on the following assumptions.
\begin{enumerate}[label=(\roman*),itemsep=2pt]
\item Arithmetic in the prime field $\F_p$, where $p=\Theta(\sqrt{\kbar})$ is chosen in~\cref{sec:testdesign}, takes $\bigO(1)$ time. The inverses of all nonzero field elements are stored in a precomputed table of $\bigO(p)$ words, so that each inversion amounts to one lookup.
\item The affine coefficients and the recorded query outcomes are stored in arrays indexed by their stage and by the relevant block or grid coordinates. Each coefficient and each outcome is therefore accessible in $\bigO(1)$ time.
\item Candidate sets support insertion, deletion, and membership queries in $\bigO(1)$ amortized time. Every unordered pair is represented canonically as $(u,v)$ with $u<v$, and inserting a pair already present leaves the set unchanged. Enumerating or copying a set of $m$ pairs takes $\bigO(m+1)$ time.
\end{enumerate}

Without loss of generality, assume that $n$ is a power of two. Otherwise, add isolated vertices up to $N:=2^{\lceil\log n\rceil}<2n$ and use the parameter $\kbar_N:=q\binom N2=\Theta(\kbar)$, restricting queries and outputs to the original vertices. Coupling the padded graph with an Erd\H{o}s--R\'enyi extension of $G$ on $[N]$, together with edge preservation and monotonicity of the candidate sets, transfers the guarantees with unchanged asymptotic bounds.
\subsection{Level Graphs}
\label{sec:blocks}

Similar to~\cite{TS25}, our algorithms work on groups of nodes, called blocks, arranged in a binary hierarchy. At each level, all vertices within a block are included together in each query. The decoder maintains candidate pairs of blocks and refines them as each block is split into two children. Starting from a coarse partition, this process continues until every block contains a single vertex. We define the hierarchy and its associated level graphs as follows.

Set
\begin{equation}
g_0:=2^{\lceil\log\sqrt{\kbar}\rceil},
\qquad
g_j:=2^jg_0\quad(0\le j\le L),
\qquad
L:=\log(n/g_0).
\label{eq:g-levels}
\end{equation}
Under~\eqref{eq:regime}, for all sufficiently large $n$, the number $L$ is a positive integer and
\begin{equation}
\sqrt{\kbar}\le g_0<2\sqrt{\kbar}<n,
\qquad g_L=n.
\label{eq:g0-range}
\end{equation}
At level $j$, partition the vertex set into $g_j$ blocks
\[
\cG_u^{(j)}:=\left\{\frac{n(u-1)}{g_j}+1,\ldots,\frac{nu}{g_j}\right\},
\qquad u\in[g_j].
\]
For $j<L$, each block $\cG_u^{(j)}$ is the disjoint union of its children $\cG_{2u-1}^{(j+1)}$ and $\cG_{2u}^{(j+1)}$. Thus the partitions are nested: vertices in the same block at one level also lie in the same block at every earlier level. When considering a single level, we omit the superscript, write $g$ for the number of blocks, and identify the blocks with $[g]$.

The level graph records which pairs of blocks contain an edge in their union.

\begin{definition}
\label{def:defective}
At a level with $g$ blocks, a block $u$ is \emph{defective} if $\cG_u$ contains an edge of $G$. An unordered pair $\{u,v\}\in\binom{[g]}2$ is \emph{defective} if $\cG_u\cup\cG_v$ contains an edge of $G$. A block or pair that is not defective is called \emph{non-defective}. The \emph{level graph} is the simple graph $G_g:=([g],E_g)$, where
\[
E_g:=\left\{\{u,v\}\in\binom{[g]}2:
\cG_u\cup\cG_v\text{ contains an edge of }G\right\}.
\]
\end{definition}

For each $g\in\{g_0,\ldots,g_L\}$, define
\begin{align}
H_g&:=\#\{\{x,y\}\in E:x,y\text{ lie in the same block}\},
\label{eq:Hg-def}\\
D_g&:=|E_g|,
\label{eq:Dg-def}\\
\Delta_g^{\times}&:=\max_{u\in[g]}
\#\{e\in E:|e\cap\cG_u|=1\}.
\label{eq:cross-def}
\end{align}
Here $H_g$ counts internal edges, $D_g$ counts defective block pairs, and $\Delta_g^{\times}$ is the largest number of edges of $G$ crossing the boundary of a single block. An internally defective block is adjacent to every other block in $G_g$, so $D_g$ also counts pairs with no edge crossing between their two blocks. The number of other blocks joined to any fixed block by crossing edges of $G$ is at most $\Delta_g^{\times}$.

\subsection{Typical Graphs}
\label{sec:structure}

The analysis of our algorithms requires bounds on the edge count, internal edges, and crossing edges of the graph. We collect these properties in a set of typical graphs and show that an Erd\H{o}s--R\'enyi graph belongs to this set with probability tending to one. We first record a deterministic bound that relates the number of defective pairs to the edge count and the number of internal edges.

\begin{fact}
\label{lem:Dg-det}
For every graph and every level with $g$ blocks, we have
\begin{equation}
D_g\le |E|+gH_g.
\label{eq:Dg-bound}
\end{equation}
\end{fact}

We now define the typical set used in the analysis.

\begin{definition}[$\epsilon_n$-typical graphs]
\label{def:typical}
Let $(\epsilon_n)$ be a sequence of positive numbers with $\epsilon_n\to0$. We define $\cT(\epsilon_n)$ as the set of graphs $G=([n],E)$ satisfying the following conditions:
\begin{enumerate}[label=(\roman*),itemsep=2pt]
\item The number of edges satisfies
\begin{equation}
(1-\epsilon_n)\kbar\le |E|\le(1+\epsilon_n)\kbar.
\label{eq:struct-edges}
\end{equation}
\item The number of internal edges at the base level satisfies
\begin{equation}
H_{g_0}\le3\sqrt{\kbar}.
\label{eq:struct-internal}
\end{equation}
\item The crossing-edge counts across all levels satisfy
\begin{equation}
\max_{0\le j\le L}\Delta_{g_j}^{\times}\le24\sqrt{\kbar}.
\label{eq:struct-cross}
\end{equation}
\item The weighted sum of internal-edge counts satisfies
\begin{equation}
\sum_{j=0}^{L-1}g_jH_{g_j}\le 3\kbar\log n.
\label{eq:struct-weighted}
\end{equation}
\end{enumerate}
A graph in $\cT(\epsilon_n)$ is called \emph{$\epsilon_n$-typical}.
\end{definition}

Every typical graph satisfies $|E|\le2\kbar$ for all sufficiently large $n$. Summing~\eqref{eq:Dg-bound} over the non-final levels and using $L\le\log n$ gives
\begin{equation}
\sum_{j=0}^{L-1}D_{g_j}
\le L|E|+\sum_{j=0}^{L-1}g_jH_{g_j}
\le 5\kbar\log n \,.
\
\label{eq:sum-Dg}
\end{equation}
This bound controls the total contribution of defective pairs to the decoding workload. The following lemma shows that the typical set has probability tending to one.

\begin{lemma}
\label{lem:structure}
For $G\sim \ER(n,q)$, there exists a sequence $\epsilon_n\to0$ such that
\[
\Pp[G\in\cT(\epsilon_n)]=1-o(1).
\]
\end{lemma}

The proof combines standard concentration inequalities with technical estimates across the block hierarchy. We provide the full proof in~\cref{app:typical}.

Table~\ref{tab:notation} summarizes the main notation used throughout the paper. 

\begin{table}[!htbp]
\centering
\small
\renewcommand{\arraystretch}{1.25}
\setlength{\tabcolsep}{5pt}
\begin{tabularx}{\textwidth}{|>{\raggedright\arraybackslash}p{0.24\textwidth}|>{\raggedright\arraybackslash}X|}
\hline
\textbf{Notation} & \textbf{Description} \\
\hline
\hline
\multicolumn{2}{|l|}{\textit{Graph and hierarchy}} \\
\hline
$n,q,\kbar$ & Number of vertices, edge probability, and expected number of edges $\kbar=q\binom{n}{2}$. \\
\hline
$g_0,g_j$ & Number of blocks at the base level and at level $j$, where $g_j=2^jg_0$. \\
\hline
$\cG_u^{(j)}$ & The $u$-th block at level $j$. \\
\hline
$H_g$ & Number of edges whose two endpoints lie in the same block, at the level with $g$ blocks. \\
\hline
$G_g,D_g$ & Level graph induced on the $g$ blocks, and its number of edges, i.e., the number of defective block pairs. \\
\hline
$\Delta_g^\times$ & Maximum, over all $g$ blocks, of the number of edges having exactly one endpoint in a given block. \\
\hline
$\cT(\epsilon_n)$ & Set of typical graphs (see~\cref{def:typical}). \\
\hline
\hline
\multicolumn{2}{|l|}{\textit{Affine query design}} \\
\hline
$p,\F_p$ & Prime of order $\Theta(\sqrt{\kbar})$, and the finite field with $p$ elements. \\
\hline
$P_u(r)=a_ur+b_u$ & Random affine map assigned to block $u$ in the current query family. \\
\hline
$S_{r,s},Y_{r,s}$ & Query consisting of the blocks with $P_u(r)=s$, and its binary outcome. \\
\hline
$\cC_e,T_e$ & For a pair $e=\{u,v\}$, the event $a_u\ne a_v$, and the collision point $T_e$ defined on $\cC_e$.\\
\hline
\hline
\multicolumn{2}{|l|}{\textit{Decoding and analysis}} \\
\hline
$\cE_e,I_e,q_e$ & Event that the pair $e$ is retained, its indicator $I_e=\ind\{\cE_e\}$, and its probability $q_e=\Pp[\cE_e]$. \\
\hline
$\rho_e$ & Conditional probability that the collision query of $e$ is positive, given $\cC_e$. \\
\hline
$\Omega_{ef}$ & Separation event for the collision queries of the pairs $e$ and $f$. \\
\hline
$\mathcal A_t$ & Affine coefficients of the query family used at stage $t$. \\
\hline
$\E_t[\cdot]$ & Expectation with respect to $\mathcal A_t$, with the graph and all earlier query families held fixed. \\
\hline
$\PD_t,Z_t$ & Set of candidate pairs entering stage $t$ and its cardinality $Z_t=|\PD_t|$. \\
\hline
$X_t,m_t,\Delta_t$ & Number of non-defective candidates retained at stage $t$, its conditional mean $m_t=\E_t[X_t]$, and the centered quantity $\Delta_t=X_t-m_t$. \\
\hline
$\beta_t,A$ & Deterministic upper bound on the number of defective candidates at stage $t$, and its sum $A=\sum_{t=0}^{M-1}\beta_t$. \\
\hline
$W$ & Total number of candidates over all stages, i.e., $W=\sum_{t=0}^{M}Z_t$. \\
\hline
$S$ & Sum of the centered stage errors, i.e., $S=\sum_{t=0}^{M-1}\Delta_t$. \\
\hline
\end{tabularx}
\caption{Notation.}
\label{tab:notation}
\end{table}
\section{Affine Splitting Approach for Graph Learning}
\label{sec:design}

Our approach follows the binary splitting framework for graph learning in~\cite{TS25}, building on non-adaptive splitting for group testing~\cite{CN20,PS20}. We use a different test design based on random affine maps over a finite field. For two blocks with distinct slopes, the intersection of their affine lines identifies their unique shared test in $\bigO(1)$ time. This removes the need to scan $\Theta(\sqrt{\kbar})$ repetitions per candidate pair, which is the main bottleneck in the basic decoder of~\cite{TS25}.

\subsection{Testing Procedure}
\label{sec:testdesign}

Each affine family assigns one random line to every block and one query to every point of a finite grid. Fix $C_1:=2000$ and let $p$ be the smallest prime satisfying $p\ge\lceil C_1\sqrt{\kbar}\rceil$. By Bertrand's postulate, for all sufficiently large $n$,
\begin{equation}
C_1\sqrt{\kbar}\le p<2\lceil C_1\sqrt{\kbar}\rceil
\le2C_1\sqrt{\kbar}+2.
\label{eq:p-range}
\end{equation}
In particular, $p=\Theta(\sqrt{\kbar})$ and $p^2=\Theta(\kbar)$. Any larger fixed value of $C_1$ is likewise admissible.

At a level with $g$ blocks, sample independently
\[
(a_u,b_u)\sim\operatorname{Unif}(\F_p^2),
\qquad
P_u(r):=a_ur+b_u,
\qquad u\in[g],\quad r\in\F_p,
\]
all arithmetic in these expressions being carried out in $\F_p$. For each grid point $(r,s)\in\F_p^2$, issue the query
\begin{equation}
S_{r,s}:=\bigcup_{u:P_u(r)=s}\cG_u.
\label{eq:affine-query}
\end{equation}
Each block participates in the $p$ tests lying on its line, and each family uses exactly $p^2$ queries. One independent family is sampled at every non-final level $j\in\{0,\ldots,L-1\}$. At the singleton level, we sample
\begin{equation}
R:=\log n
\label{eq:final-rounds}
\end{equation}
independent affine families. All families, including those belonging to different levels and to different final rounds, are mutually independent. Superscripts $(j)$ identify non-final levels, whereas superscripts $[\tau]$ identify final rounds.

\paragraph{Collision structure.}
For distinct blocks $u,v$, define
\[
\cC_{uv}:=\{a_u\ne a_v\}.
\]
On this event, the equation $P_u(r)=P_v(r)$ admits the unique solution
\begin{equation}
r_{uv}^*:=(b_v-b_u)(a_u-a_v)^{-1},
\qquad
s_{uv}^*:=a_ur_{uv}^*+b_u.
\label{eq:collision-point}
\end{equation}
The blocks $u$ and $v$ therefore occur together in exactly one query, indexed by $T_{uv}:=(r_{uv}^*,s_{uv}^*)$. If $a_u=a_v$, their lines either coincide, which occurs precisely when $b_u=b_v$, or do not meet at all. On this complementary event, set $T_{uv}:=\dagger$, where $\dagger$ denotes a symbol outside $\F_p^2$; for definiteness, set $r_{uv}^*=s_{uv}^*=0$ there. These assigned coordinates are invoked only inside events that also require the relevant collisions to occur.

For $T=(r,s)$, write $Y_T=Y_{r,s}:=Y(S_{r,s})$, and abbreviate the membership event $\{P_w(r)=s\}$ by $\{w\in T\}$, with the conventions $Y_{\dagger}:=0$ and $\{w\in\dagger\}:=\varnothing$. Under the computational assumptions, evaluating~\eqref{eq:collision-point} and reading the corresponding outcome take $\bigO(1)$ time. A negative collision test consequently rules out a candidate pair by means of a single lookup, in place of the scan through $\Theta(\sqrt{\kbar})$ repetitions required by the basic graph-splitting construction of~\cite{TS25}.

The collision calculation is illustrated by the following example, whose small field serves purely illustrative purposes.

\begin{example}
\label{ex:affine}
Consider a level with $g=8$ blocks and take $p=5$. Suppose two blocks carry the maps $P_u(r)=2r+1$ and $P_v(r)=4r+3$ over $\F_5$. Their values are
\[
\begin{array}{c|ccccc}
r&0&1&2&3&4\\\hline
P_u(r)&1&3&0&2&4\\
P_v(r)&3&2&1&0&4
\end{array}
\]
and coincide only at $r=4$. Indeed,
\[
r_{uv}^*=(3-1)(2-4)^{-1}=2\cdot3^{-1}=4,
\qquad s_{uv}^*=4.
\]
The pair thus appears together in the unique test indexed by $(4,4)$, as displayed in~\cref{fig:affine}.
\end{example}

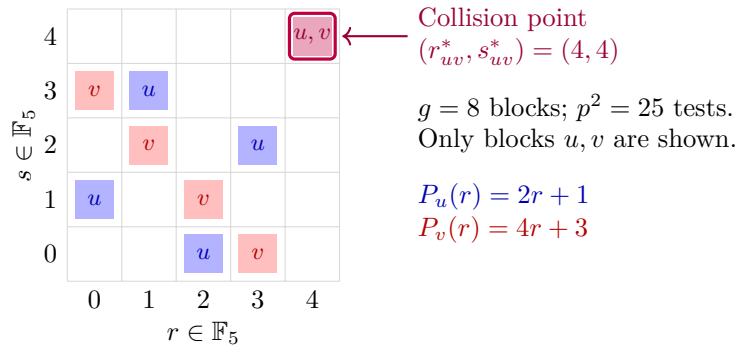
\begin{figure}[!htbp]
\centering
\begin{tikzpicture}[scale=0.72,every node/.style={font=\small}]
\draw[step=1,gray!30,thin] (0,0) grid (5,5);
\foreach \r in {0,...,4} {\node[below] at (\r+0.5,0) {$\r$};}
\foreach \s in {0,...,4} {\node[left] at (0,\s+0.5) {$\s$};}
\node[below] at (2.5,-0.55) {$r\in\F_5$};
\node[rotate=90] at (-0.8,2.5) {$s\in\F_5$};
\foreach \r/\s in {0/1,1/3,2/0,3/2} {
  \fill[blue!30] (\r+0.15,\s+0.15) rectangle (\r+0.85,\s+0.85);
  \node[blue!70!black,font=\footnotesize\bfseries] at (\r+0.5,\s+0.5) {$u$};}
\foreach \r/\s in {0/3,1/2,2/1,3/0} {
  \fill[red!25] (\r+0.15,\s+0.15) rectangle (\r+0.85,\s+0.85);
  \node[red!70!black,font=\footnotesize\bfseries] at (\r+0.5,\s+0.5) {$v$};}
\fill[purple!35] (4.15,4.15) rectangle (4.85,4.85);
\node[purple!80!black,font=\footnotesize\bfseries] at (4.5,4.5) {$u,v$};
\draw[purple,very thick,rounded corners=2pt] (4.08,4.08) rectangle (4.92,4.92);
\draw[->,thick,purple!80!black] (6.25,4.5) -- (5.05,4.5);
\node[right,align=left,purple!80!black] at (6.25,4.5)
  {Collision point\\$(r_{uv}^*,s_{uv}^*)=(4,4)$};
\node[right,align=left] at (6.25,2.9)
  {$g=8$ blocks; $p^2=25$ tests.\\Only blocks $u,v$ are shown.};
\node[right,blue!70!black] at (6.25,1.55) {$P_u(r)=2r+1$};
\node[right,red!70!black] at (6.25,0.95) {$P_v(r)=4r+3$};
\end{tikzpicture}
\caption{An affine test assignments over $\F_5$.}
\label{fig:affine}
\end{figure}

The complete non-adaptive test design is presented in~\cref{alg:test}.

\begin{algorithm}[!htbp]
\caption{Affine testing procedure}
\label{alg:test}
\begin{algorithmic}[1]
\REQUIRE $n$, $\kbar$, $C_1=2000$.
\STATE Define $g_0,\ldots,g_L$ by~\eqref{eq:g-levels} and $R$ by~\eqref{eq:final-rounds}.
\STATE $p\gets$ smallest prime at least $\lceil C_1\sqrt{\kbar}\rceil$.
\FOR{$j=0,\ldots,L-1$}
  \STATE Sample independent uniform $(a_u^{(j)},b_u^{(j)})\in\F_p^2$ for $u\in[g_j]$.
  \STATE Record $Y_{r,s}^{(j)}\gets Y\!\left(\bigcup_{u:a_u^{(j)}r+b_u^{(j)}=s}\cG_u^{(j)}\right)$ for all $(r,s)\in\F_p^2$.
\ENDFOR
\FOR{$\tau=1,\ldots,R$}
  \STATE Sample independent uniform $(a_u^{[\tau]},b_u^{[\tau]})\in\F_p^2$ for $u\in[n]$.
  \STATE Record $Y_{r,s}^{[\tau]}\gets Y\!\left(\{u:a_u^{[\tau]}r+b_u^{[\tau]}=s\}\right)$ for all $(r,s)\in\F_p^2$.
\ENDFOR
\end{algorithmic}
\end{algorithm}

\paragraph{Distributional properties of the affine family.}
\label{sec:dist}
We record the one-point, two-point, and collision laws used in the analysis. In the remainder of this subsection, the partition and the graph are fixed, and probabilities refer to one affine family alone. Each law is a consequence of the independent uniform coefficient pairs and applies to every non-final level and every final round.

\begin{proposition}
\label{prop:one-point}
For every fixed $r\in\F_p$, the variables $(P_u(r))_{u\in[g]}$ are independent and uniform on $\F_p$.
\end{proposition}

\begin{proof}
For every $s\in\F_p$ and every $a\in\F_p$, exactly one value $b=s-ar$ satisfies $ar+b=s$. Each value therefore has exactly $p$ preimages under $(a,b)\mapsto ar+b$, which proves uniformity. Independence follows since each image depends only on the coefficient pair of its own block.
\end{proof}


\begin{proposition}
\label{prop:two-point}
For every block $w$ and distinct $r,r'\in\F_p$, the vector $(P_w(r),P_w(r'))$ is uniform on $\F_p^2$.
\end{proposition}

\begin{proof}
The linear map $(a,b)\mapsto(ar+b,ar'+b)$ has matrix $\begin{pmatrix}r&1\\r'&1\end{pmatrix}$
with nonzero determinant $r-r'$. It is therefore a bijection of $\F_p^2$ and preserves the uniform law.
\end{proof}

The next proposition describes the complete membership distribution of the test selected by a collision.

\begin{proposition}
\label{prop:composition}
Let $u,v$ be distinct blocks. Condition on $\cC_{uv}$ and on the collision point $T_{uv}=(r,s)$ from~\eqref{eq:collision-point}, assuming that this event has positive probability. Then:
\begin{enumerate}[label=(\roman*),itemsep=2pt]
\item both $u$ and $v$ belong to the test $(r,s)$;
\item each block outside $\{u,v\}$ belongs to this test independently with probability $1/p$;
\item the conditional joint distribution of these membership indicators does not depend on $(r,s)$.
\end{enumerate}
\end{proposition}

\begin{proof}
On the conditioning event, $P_u(r)=P_v(r)=s$, which proves~(i). This event depends only on the coefficients of $u$ and $v$. The coefficient pairs of all other blocks therefore remain independent and uniform on $\F_p^2$. By~\cref{prop:one-point}, their values at $r$ are independent and uniform on $\F_p$. Each value equals $s$ with probability $1/p$, proving~(ii). Thus, the membership indicators are independent Bernoulli variables with parameter $1/p$. Their joint distribution is the same for every $(r,s)$, which proves~(iii).
\end{proof}

Recall that, on $\cC_{uv}=\{a_u\ne a_v\}$, $r_{uv}^*$ is the first coordinate of the collision point defined in~\eqref{eq:collision-point}. These coordinates satisfy the following properties.
\begin{proposition}
\label{prop:first_cordinate}
The first coordinates $r_{uv}^*$ defined in~\eqref{eq:collision-point} satisfy the following properties.
\begin{enumerate}[label=(\roman*),itemsep=2pt]
\item For distinct blocks $u,v$, the first coordinate $r_{uv}^*$ is uniform on $\F_p$ conditional on $\cC_{uv}$.
\item For disjoint pairs $\{u,v\}$ and $\{u',v'\}$, the first coordinates $r_{uv}^*$ and $r_{u'v'}^*$ are independent and uniform on $\F_p$ conditional on $\cC_{uv}\cap\cC_{u'v'}$.
\item For pairs $\{u,v\}$ and $\{u,v'\}$ with $v\ne v'$, the first coordinates $r_{uv}^*$ and $r_{uv'}^*$ are independent and uniform on $\F_p$ conditional on $\cC_{uv}\cap\cC_{uv'}$ and the coefficients $(a_u,b_u)$ of the shared block.
\end{enumerate}
In both~(ii) and~(iii), the probability that both collision events hold and the two first coordinates are equal is at most $1/p$.
\end{proposition}

\begin{proof}
For~(i), fix $a_u,b_u,a_v$ with $a_u\ne a_v$. The coefficient $b_v$ remains uniform on $\F_p$, and
\[
r_{uv}^*=(b_v-b_u)(a_u-a_v)^{-1}.
\]
Since $a_u-a_v\ne0$, each value of $r_{uv}^*$ corresponds to exactly one value of $b_v$. Thus, $r_{uv}^*$ is uniform on $\F_p$. This holds for every choice of the fixed coefficients, so averaging over them proves~(i).

For~(ii), the two pairs involve disjoint sets of independent coefficients. Conditioning on $\cC_{uv}\cap\cC_{u'v'}$ preserves their independence because each event depends only on the coefficients of its own pair. Applying~(i) to each pair gives the claim.

For~(iii), fix $(a_u,b_u)$. The coefficient pairs $(a_v,b_v)$ and $(a_{v'},b_{v'})$ remain independent. The conditions $a_v\ne a_u$ and $a_{v'}\ne a_u$ restrict these pairs separately, so they remain independent after conditioning. The argument in~(i) shows that each first coordinate is uniform on $\F_p$, proving~(iii).

Two independent uniform elements of $\F_p$ are equal with probability $1/p$. In case~(iii), this probability is the same for every choice of the shared coefficients, so it remains $1/p$ after averaging over them. In either case, multiplying by the probability that both collision events hold gives the final bound.
\end{proof}


\subsection{Decoding Procedure}
\label{sec:decoder}

Given the affine coefficients and query outcomes from~\cref{alg:test}, the decoder is deterministic. It follows the binary refinement framework of~\cite{TS25}, maintaining candidate pairs of distinct blocks and refining them until the blocks are singletons. For $0\le j\le L$, let $\PD_j$ denote the candidate set entering level $j$. Initially, all pairs are candidates, so $\PD_0:=\binom{[g_0]}2$.

At each non-final level $j$, the decoder examines every pair $\{u,v\}\in\PD_j$. If $a_u^{(j)}\ne a_v^{(j)}$, it computes the collision point $T_{uv}^{(j)}$ using~\eqref{eq:collision-point} and discards the pair if $Y_{T_{uv}^{(j)}}^{(j)}=0$. Pairs with equal slopes are retained automatically. Write $u_L,u_R$ and $v_L,v_R$ for the children of $u$ and $v$, respectively. Each retained pair contributes the following six pairs to $\PD_{j+1}$:
\begin{equation}
\begin{gathered}
\{u_L,v_L\},\quad \{u_L,v_R\}, \quad 
\{u_R,v_L\},\quad \{u_R,v_R\}, \quad 
\{u_L,u_R\},\quad \{v_L,v_R\}.
\end{gathered}
\label{eq:children}
\end{equation}
As in~\cite{TS25}, the four cross pairs preserve edges between the parent blocks, while the two sibling pairs preserve edges within them. Each child pair is stored only once, even if it is generated by several retained pairs.

At level $L$, the blocks are singletons. The decoder applies the same retention rule for $R$ final rounds, using the independent affine family assigned to each round and generating no further children. Let $\PD_{L+\tau}$ denote the candidate set after round $\tau$, where $1\le\tau\le R$. These sets satisfy $\PD_{L+\tau}\subseteq\PD_{L+\tau-1}$, and the decoder returns $\Ehat:=\PD_{L+R}$.

For the analysis, a \emph{stage} means either a non-final level or a final round. There are $M:=L+R$ stages, indexed by $t=0,\ldots,M-1$. Stage $t=j$ processes non-final level $j$, and stage $t=L+\tau-1$ processes final round $\tau$. Thus $\PD_t$ is the candidate set entering stage $t$, and $\PD_M$ is the output. The complete procedure is given in~\cref{alg:decoder}; superscripts $(j)$ and $[\tau]$ identify the affine families used at level $j$ and final round $\tau$, respectively.

\begin{algorithm}[!htbp]
\caption{Affine splitting decoder}
\label{alg:decoder}
\begin{algorithmic}[1]
\REQUIRE $n$, $\kbar$, $R$, and the affine coefficients and query outcomes from~\cref{alg:test}.
\STATE $\PD_0\gets\binom{[g_0]}2$.
\FOR{$j=0,\ldots,L-1$}
  \STATE $\PD_{j+1}\gets\varnothing$.
  \FOR{$\{u,v\}\in\PD_j$}
    \IF{$a_u^{(j)}=a_v^{(j)}$}
      \STATE Retain $\{u,v\}$.
    \ELSE
      \STATE Compute $T_{uv}^{(j)}$ using~\eqref{eq:collision-point}.
      \STATE Retain $\{u,v\}$ if and only if $Y_{T_{uv}^{(j)}}^{(j)}=1$.
    \ENDIF
    \IF{$\{u,v\}$ is retained}
      \STATE Insert its six child pairs from~\eqref{eq:children} into $\PD_{j+1}$.
    \ENDIF
  \ENDFOR
\ENDFOR
\FOR{$\tau=1,\ldots,R$}
  \STATE $\PD_{L+\tau}\gets\PD_{L+\tau-1}$.
  \FOR{$\{x,y\}\in\PD_{L+\tau-1}$}
    \IF{$a_x^{[\tau]}\ne a_y^{[\tau]}$}
      \STATE Compute $T_{xy}^{[\tau]}$ using~\eqref{eq:collision-point}.
      \STATE Delete $\{x,y\}$ from $\PD_{L+\tau}$ if $Y_{T_{xy}^{[\tau]}}^{[\tau]}=0$.
    \ENDIF
  \ENDFOR
\ENDFOR
\RETURN $\Ehat\gets\PD_M$.
\end{algorithmic}
\end{algorithm}

Processing a candidate requires one slope comparison, at most one field inversion, a constant number of further field operations, at most one outcome lookup, and at most six set updates. Under our computational assumptions, this takes a constant amortized number of operations per candidate. Including initialization, set copying, and output, the total decoding time is therefore $\bigO(M+\sum_{t=0}^{M}|\PD_t|)$. We bound this quantity in~\cref{sec:analysis}.

\section{Algorithmic Guarantees}
\label{sec:analysis}

This section proves the guarantees on the number of queries, recovery, and decoding time. Using the distributional properties in~\cref{sec:dist}, we first bound the probability that a candidate survives one stage and the covariance between survival indicators. We then use these bounds to control the total number of candidates processed by the decoder. Finally, we show that every true edge is preserved and that the final rounds remove all non-edges with high probability.

The guarantees are summarized in the following theorem, whose running-time statement is understood under the computational conventions of~\cref{sec:problem}.

\begin{theorem}
\label{thm:main}
Let $G\sim \ER(n,q)$, where $\kbar=q\binom n2$ satisfies $\kbar\to\infty$ and $\kbar=o(n^2)$. Consider the affine test design in Algorithm~\ref{alg:test} and the splitting decoder in Algorithm~\ref{alg:decoder} with $C_1=2000$ and $R=\log n$. Then the following guarantees hold as $n\to\infty$:
\begin{enumerate}[label=(\roman*),itemsep=3pt]
\item the test design uses $\bigO(\kbar\log n)$ non-adaptive queries;
\item the decoder recovers $E$ exactly with probability $1-o(1)$;
\item the decoding time is $\bigO(\kbar\log n)$ with probability $1-o(1)$.
\end{enumerate}
\end{theorem}

By~\cref{lem:structure}, $G\in\cT(\epsilon_n)$ with probability $1-o(1)$. It therefore suffices to prove the required guarantees uniformly over graphs in this set. For the intermediate analysis, we fix such a graph $G \in\cT(\epsilon_n)$, so the only randomness comes from the affine families, and each query outcome is determined by the family used at its stage. In~\cref{sec:stage}, the symbols $\Pp,\E[\cdot],\Var$, and $\Cov$ refer to a single family, as in~\cref{sec:dist}; in~\cref{sec:global}, they refer to all stage families jointly. We combine these bounds with~\cref{lem:structure} in~\cref{sec:main}.

\subsection{One-Stage Survival and Covariance}
\label{sec:stage}

Fix a stage with $g$ blocks and a non-defective pair $e=\{u,v\}$. Let $\cC_e:=\cC_{uv}=\{a_u\ne a_v\}$ be the event that the two blocks have distinct slopes. On this event, $T_e:=T_{uv}$ is their collision point and identifies their shared test; write $r_e^*:=r_{uv}^*$ and $s_e^*:=s_{uv}^*$ for its coordinates. The pair is retained if its slopes are equal or its collision test is positive. Accordingly, define the retention event $\cE_e$ and its indicator $I_e$ by
\begin{equation}
\cE_e:=\cC_e^c\cup\bigl(\cC_e\cap\{Y_{T_e}=1\}\bigr),
\qquad I_e:=\ind\{\cE_e\}.
\label{eq:survival-event}
\end{equation}
Thus $I_e=1$ if $e$ is retained and $I_e=0$ otherwise. The next lemma bounds the retention probability uniformly by a constant small enough to account for the six child pairs generated by each retained pair.

\begin{lemma}
\label{lem:survival}
For every fixed $G\in\cT(\epsilon_n)$ and all sufficiently large $n$, every non-defective pair $e$ at every non-final level and every final round satisfies
\begin{equation}
q_e:=\Pp[\cE_e]\le \frac1{24}.
\label{eq:alpha}
\end{equation}
\end{lemma}

\begin{proof}
The slopes of the two blocks are independent and uniform, whence $\Pp[\mathcal{C}_e^c]=1/p$ and
\begin{equation}
q_e=\frac1p+\left(1-\frac1p\right)\rho_e,
\qquad
\rho_e:=\Pp[Y_{T_e}=1\mid\mathcal{C}_e].
\label{eq:q-decomp}
\end{equation}
By~\cref{prop:composition}, the conditional positivity probability given $\mathcal{C}_e$ and $T_e=(r,s)$ is the same for every admissible $(r,s)$ and hence equals $\rho_e$. Fix one such point and let
\[
\mathcal{B}:=\{w\in[g]\setminus e:\ w\in T_e\}
\]
denote the set of non-special blocks belonging to the test, so that the queried vertex set is $\bigcup_{x\in e\cup\mathcal{B}}\mathcal{G}_x$. By~\cref{prop:composition}(ii), the blocks of $[g]\setminus e$ lie in $\mathcal{B}$ independently with probability $1/p$.

Since $e$ is non-defective, neither block of $e$ contains an internal edge and no edge joins the two blocks. Every edge witnessing a positive outcome therefore meets at least one block of $\mathcal{B}$, and its two endpoint blocks determine which of the following three events occurs:
\begin{align}
\mathcal{F}_1&:=\{\exists\,w\in\mathcal{B}:\ \mathcal{G}_w\text{ contains an edge}\},
\label{eq:event-A1}\\
\mathcal{F}_2&:=\{\exists\,w\in\mathcal{B},\ x\in e:\ \text{some edge joins }\mathcal{G}_w\text{ and }\mathcal{G}_x\},
\label{eq:event-A2}\\
\mathcal{F}_3&:=\{\exists\,w,w'\in\mathcal{B},\ w\ne w':\ \text{some edge joins }\mathcal{G}_w\text{ and }\mathcal{G}_{w'}\}.
\label{eq:event-A3}
\end{align}
These cases exhaust the possibilities, an edge whose endpoint blocks both lie in $e$ being excluded by non-defectiveness. Consequently,
\begin{equation}
\{Y_{T_e}=1\}\subseteq\mathcal{F}_1\cup\mathcal{F}_2\cup\mathcal{F}_3 .
\label{eq:event-cover}
\end{equation}
Each event is now bounded separately, all three estimates being uniform over the admissible collision points.

\emph{The event $\mathcal{F}_1$.}
Distinct internally defective blocks contain distinct internal edges, so that the number of blocks $w$ with $\mathcal{G}_w$ containing an edge is at most $H_g$; by nestedness and~\eqref{eq:struct-internal}, this is at most $H_{g_0}\le3\sqrt{\kbar}$, and at the singleton level it is zero. Each such block lies in $\mathcal{B}$ with probability $1/p$, whence
\begin{equation}
\Pp[\mathcal{F}_1]\le\frac{3\sqrt{\kbar}}p .
\label{eq:bound-A1}
\end{equation}

\emph{The event $\mathcal{F}_2$.}
A block $w$ joined to $\mathcal{G}_u$ or to $\mathcal{G}_v$ requires an edge crossing the boundary of $\mathcal{G}_u$ or of $\mathcal{G}_v$, and distinct such blocks require distinct crossing edges; there are therefore at most $2\Delta_g^{\times}\le48\sqrt{\kbar}$ of them by~\eqref{eq:struct-cross}. The same union bound gives
\begin{equation}
\Pp[\mathcal{F}_2]\le\frac{48\sqrt{\kbar}}p .
\label{eq:bound-A2}
\end{equation}

\emph{The event $\mathcal{F}_3$.}
An edge witnesses $\mathcal{F}_3$ only if its two endpoint blocks are distinct, lie outside $e$, and both belong to $\mathcal{B}$, which by independence has probability $1/p^2$. Since at most $|E|\le2\kbar$ edges qualify, by~\eqref{eq:struct-edges},
\begin{equation}
\Pp[\mathcal{F}_3]\le\frac{2\kbar}{p^2}.
\label{eq:bound-A3}
\end{equation}

Combining~\eqref{eq:event-cover} with~\eqref{eq:bound-A1}--\eqref{eq:bound-A3} yields
\[
\rho_e\le\frac{51\sqrt{\kbar}}p+\frac{2\kbar}{p^2}
\le\frac{51}{C_1}+\frac2{C_1^2},
\]
the last inequality following from $p\ge C_1\sqrt{\kbar}$. For sufficiently large $n$ one has $\kbar\ge1$ and hence $p\ge C_1$. Substitution into~\eqref{eq:q-decomp} yields
\[
q_e\le\frac1p+\rho_e
\le\frac{52}{C_1}+\frac2{C_1^2}
\le\frac1{24}.
\qedhere
\]
\end{proof}

\paragraph{Covariance within one family.}
Distinct survival indicators are generated by the same affine family and need not be independent. Their covariance is bounded from above by separating the two collision tests. For every non-defective pair $f$, the preceding proof together with~\cref{prop:composition} gives
\begin{equation}
\rho_f
=\Pp[Y_{T_f}=1\mid\cC_f,\ T_f=(r',s')]
\le q_f
\label{eq:rho-le-q}
\end{equation}
for every admissible collision point. Equivalently, $\rho_f$ is the positivity probability of a test that contains $f$ and includes every other block independently with probability $1/p$.

The following event requires the two collision first coordinates to be distinct and excludes the unshared blocks of each pair from the test of the other pair.

\begin{definition}
\label{def:separation}
For distinct non-defective pairs $e,f$, define
\begin{equation}
\Omega_{ef}:=
\cC_e\cap\cC_f\cap\{r_e^*\ne r_f^*\}
\cap\bigcap_{w\in e\setminus f}\{w\notin T_f\}
\cap\bigcap_{w\in f\setminus e}\{w\notin T_e\}.
\label{eq:sep-event}
\end{equation}
\end{definition}

The event $\Omega_{ef}$ requires distinct slopes within each pair, different first coordinates for the two collision points, and exclusion of every block belonging to only one pair from the other pair's test. If $e$ and $f$ share a block, that block belongs to both collision tests and is not subject to exclusion. The next proposition bounds $\Pp[\Omega_{ef}^c]$, the probability that at least one of these conditions fails. Here $e\triangle f:=(e\setminus f)\cup(f\setminus e)$ is the set of blocks belonging to exactly one of the two pairs.

\begin{proposition}
\label{prop:sep-prob}
For any two distinct non-defective pairs $e,f$, the separation event $\Omega_{ef}$ defined in~\eqref{eq:sep-event} satisfies
\begin{equation}
\Pp[\Omega_{ef}^c]
\le\frac{3+|e\triangle f|}{p}\le\frac7p.
\label{eq:omega-bad}
\end{equation}
\end{proposition}

On the separation event, the joint survival probability admits the following product bound, whose proof conditions explicitly on the special coefficients and on the complete membership pattern of the first test.

\begin{proposition}
\label{prop:product}
For any two distinct non-defective pairs $e,f$, with retention events defined in~\eqref{eq:survival-event} and separation event defined in~\eqref{eq:sep-event}, we have
\begin{equation}
\Pp[\cE_e\cap\cE_f\cap\Omega_{ef}]\le q_eq_f.
\label{eq:product-bound}
\end{equation}
\end{proposition}

Both propositions are proved in~\cref{app:covariance}. Taken together, they yield at once the covariance estimate required by the second-moment analysis.

\begin{lemma}
\label{lem:covariance}
For distinct unordered non-defective pairs $e,f$ processed under one affine family, one has
\begin{equation}
\Cov(I_e,I_f)\le\frac7p \, ,
\label{eq:cov}
\end{equation}
where $I_e:=\ind\{\cE_e\}.$
\end{lemma}

\begin{proof}
Since $I_e$ and $I_f$ are indicators with means $q_e$ and $q_f$, respectively. We have
\[
\Cov(I_e,I_f)=\Pp[\cE_e\cap\cE_f]-q_eq_f.
\]
We split the joint event according to whether $\Omega_{ef}$ occurs. By~\cref{prop:product,prop:sep-prob}, one has
\begin{align*}
\Pp[\cE_e\cap\cE_f]
\le \Pp[\cE_e\cap\cE_f\cap\Omega_{ef}]+\Pp[\Omega_{ef}^c]
\le q_eq_f+\frac7p.
\end{align*}
Subtracting $q_eq_f$ gives $\Cov(I_e,I_f)\le 7/p$, as required.
\end{proof}

The one-stage analysis concludes by summing the individual survival indicators. The candidate collection appearing in the next statement is deterministic relative to the fresh family.

\begin{corollary}
\label{cor:moments}
Let $\cN$ be a deterministic collection of non-defective pairs, and let $X:=\sum_{e\in\cN}I_e$ under a fresh affine family. Then
\begin{align}
\E[X]&\le\frac{|\cN|}{24},
\label{eq:stage-mean}\\
\Var(X)&\le|\cN|+\frac7p|\cN|^2.
\label{eq:stage-var}
\end{align}
\end{corollary}

\begin{proof}
The mean bound is immediate from~\cref{lem:survival}. For the variance, expand
\[
\Var(X)=\sum_{e\in\cN}\Var(I_e)
+\sum_{\substack{e,f\in\cN\\e\ne f}}\Cov(I_e,I_f).
\]
Each diagonal term is at most one, and each off-diagonal term is at most $7/p$ by~\cref{lem:covariance}. Since there are $|\cN|(|\cN|-1)\le|\cN|^2$ ordered off-diagonal pairs, the bound~\eqref{eq:stage-var} follows.
\end{proof}

\subsection{Global Analysis of the Candidate Count}
\label{sec:global}

The stages are now combined so as to bound the total candidate workload. The argument rests on the aggregate defective-pair estimate~\eqref{eq:sum-Dg}, on the one-stage moment bounds, and on the vanishing cross moments of the centered stage errors. No separate bound on the candidate set is imposed at any individual level.

Recall that $M=L+R$ and that $\PD_t$ denotes the candidate set entering global stage $t$, with $\PD_M=\Ehat$. For $0 \le t \le M$, define 
\begin{equation}
Z_t:=|\PD_t| \, , \text{ and }
W:=\sum_{t=0}^{M}Z_t.
\label{eq:workload-def}
\end{equation}




For $0\le t<M$, let $\mathcal A_t$ denote the collection of affine coefficients used at stage $t$. These collections are mutually independent. Since the graph is fixed and the decoder is deterministic, $\PD_t$ and $Z_t$ depend only on $\mathcal A_0,\ldots,\mathcal A_{t-1}$.

We write $\Pp_t$, $\E_t[\cdot]$, and $\Var_t$ for probability, expectation, and variance conditional on these earlier families. For a random variable $Y$ determined by $\mathcal A_0,\ldots,\mathcal A_t$, independence allows us to compute
\[
\E_t[Y]=\E[Y\mid\mathcal A_0,\ldots,\mathcal A_{t-1}]
\]
by averaging over $\mathcal A_t$ alone. At $t=0$, there are no earlier families to condition on. The law of total expectation gives
\begin{equation}
\E[Y]=\E[\E_t[Y]].
\label{eq:successive-averaging}
\end{equation}

Let $X_t$ count the non-defective candidates in $\PD_t$ that are retained at stage $t$. The deterministic upper bounds on the defective candidates are
\begin{equation}
\beta_t:=
\begin{cases}
D_{g_t},&0\le t<L,\\
|E|,&L\le t<M,
\end{cases}
\qquad
A:=\sum_{t=0}^{M-1}\beta_t,
\label{eq:beta-A}
\end{equation}
the second case being justified by the fact that the defective singleton pairs at a final round are precisely the edges. For every fixed typical graph,~\eqref{eq:sum-Dg} gives
\begin{equation}
A=\sum_{j=0}^{L-1}D_{g_j}+R|E|
=\bigO(\kbar\log n),
\qquad
Z_0=\binom{g_0}2<2\kbar \, .
\label{eq:A-Z0}
\end{equation}
Here the implicit constant is absolute.

\paragraph{One-step recursion.}
At a non-final stage, at most $\beta_t$ defective candidates and $X_t$ non-defective candidates are retained, and each retained pair generates at most six children; at a final round no children are generated, so that the same upper bound remains valid. Therefore,
\begin{equation}
Z_{t+1}\le6(\beta_t+X_t),
\qquad 0\le t<M.
\label{eq:one-step-global}
\end{equation}
Fix the affine coefficients of all stages before $t$, and let $\cN_t\subseteq\PD_t$ be the set of non-defective candidates. This set is then deterministic. Since $|\cN_t|\le Z_t$,~\cref{cor:moments} yields
\begin{align}
m_t:=\E_t[X_t]&\le\frac{Z_t}{24},
\label{eq:stagewise-mean-global}\\
\Var_t(X_t)&\le Z_t+\frac7pZ_t^2.
\label{eq:stagewise-var-global}
\end{align}

Define $\Delta_t:=X_t-m_t$, the difference between $X_t$ and its mean over the current affine family. With the affine coefficients of all earlier stages fixed, we have
\begin{equation}
\E_t[\Delta_t]=0,
\qquad
\E_t[\Delta_t^2]=\Var_t(X_t)
\le Z_t+\frac7pZ_t^2.
\label{eq:centered-stage}
\end{equation}
Using $X_t=m_t+\Delta_t$ and $m_t\le Z_t/24$ in~\eqref{eq:one-step-global} gives
\begin{equation}
Z_{t+1}\le6\beta_t+\frac14Z_t+6\Delta_t.
\label{eq:recursion-centered}
\end{equation}
With the earlier affine coefficients fixed, non-defective candidates contribute at most $Z_t/4$ candidates to the next stage in expectation. We first bound the sum of the expected candidate counts using the one-stage mean bound.

\begin{lemma}
\label{lem:first-moment}
For every fixed $G\in\cT(\epsilon_n)$ and all sufficiently large $n$, we have 
\begin{equation}
\sum_{t=0}^{M}\E[Z_t]
\le\frac43(Z_0+6A)
=\bigO(\kbar\log n) \, ,
\label{eq:first-sum}
\end{equation}
where $Z_t$ and $A$ are defined in~\eqref{eq:workload-def} and~\eqref{eq:beta-A}, respectively. 
\end{lemma}

\begin{proof}
Fix the affine coefficients of all stages before $t$ and take expectation over $\mathcal A_t$ in~\eqref{eq:recursion-centered}. Since $\beta_t$ and $Z_t$ are fixed and $\E_t[\Delta_t]=0$, we obtain
\[
\E_t[Z_{t+1}]\le6\beta_t+\frac14Z_t.
\]
Averaging over the earlier families and using~\eqref{eq:successive-averaging} gives
\begin{equation}
\E[Z_{t+1}]\le6\beta_t+\frac14\E[Z_t].
\label{eq:first-step-global}
\end{equation}
Let $B_1:=\sum_{t=0}^{M}\E[Z_t]$. Summing~\eqref{eq:first-step-global} over $0\le t<M$ gives
\[
B_1-Z_0\le6A+\frac14(B_1-\E[Z_M]).
\]
Since $\E[Z_M]\ge0$, we have
\[
\frac34B_1\le Z_0+6A.
\]
Therefore $B_1\le\frac43(Z_0+6A)$, and~\eqref{eq:A-Z0} gives the stated bound.
\end{proof}

The variance bound in~\eqref{eq:centered-stage} involves $Z_t^2$. We therefore also bound the sum of the second moments of the candidate counts.

\begin{lemma}
\label{lem:L2}
There is an absolute constant $C_2$ such that, for every fixed typical graph and all sufficiently large $n$,
\begin{equation}
\sum_{t=0}^{M}\E[Z_t^2]
\le C_2(Z_0^2+A^2+M).
\label{eq:L2-sum}
\end{equation}
\end{lemma}

\begin{proof}
Fix the affine coefficients of all stages before $t$. Both sides of~\eqref{eq:one-step-global} are nonnegative. Squaring this inequality and taking expectation over $\mathcal A_t$ gives
\begin{align*}
\E_t[Z_{t+1}^2]
&\le36\E_t[(\beta_t+X_t)^2]\\
&=36(\beta_t+m_t)^2+36\Var_t(X_t)\\
&\le36\left(\beta_t+\frac{Z_t}{24}\right)^2
  +36Z_t+\frac{252}{p}Z_t^2\\
&\le72\beta_t^2+
  \left(\frac18+\frac{252}{p}\right)Z_t^2+36Z_t.
\end{align*}
The third line follows from~\eqref{eq:stagewise-mean-global}--\eqref{eq:stagewise-var-global}, and the last line uses $(x+y)^2\le2x^2+2y^2$. Since $p\to\infty$, we have $252/p\le1/8$ for all sufficiently large $n$. Also,
\[
36z\le\frac18z^2+2592\qquad(z\ge0),
\]
which follows from $(z-144)^2\ge0$. Combining these bounds, we obtain
\[
\E_t[Z_{t+1}^2]\le\frac12Z_t^2+C(\beta_t^2+1)
\]
for an absolute constant $C$. Averaging over the earlier families gives
\begin{equation}
\E[Z_{t+1}^2]\le\frac12\E[Z_t^2]+C(\beta_t^2+1).
\label{eq:L2-step-full}
\end{equation}
Let $B_2:=\sum_{t=0}^{M}\E[Z_t^2]$. Summing~\eqref{eq:L2-step-full} over $0\le t<M$ yields
\[
B_2-Z_0^2\le\frac12(B_2-\E[Z_M^2])
+C\left(\sum_{t=0}^{M-1}\beta_t^2+M\right).
\]
Using $\E[Z_M^2]\ge0$ and
\[
\sum_{t=0}^{M-1}\beta_t^2
\le\left(\sum_{t=0}^{M-1}\beta_t\right)^2=A^2,
\]
we obtain
\[
B_2\le2Z_0^2+2C(A^2+M).
\]
Choosing $C_2:=\max\{2,2C\}$ proves the result.
\end{proof}

The errors $\Delta_t$ may be dependent, but each has mean zero when the affine coefficients of all earlier stages are fixed. This makes the cross terms vanish in the second moment of their sum.

\begin{lemma}
\label{lem:orthogonality}
Let $\Delta_t=X_t-m_t$ be the centered stage error, where $m_t=\E_t[X_t]$. For $0\le s<t<M$, we have
\begin{equation}
\E[\Delta_s\Delta_t]=0.
\label{eq:orthogonality}
\end{equation}
Consequently, with $S:=\sum_{t=0}^{M-1}\Delta_t$,
\begin{equation}
\E[S]=0,
\qquad
\E[S^2]=\sum_{t=0}^{M-1}\E[\Delta_t^2].
\label{eq:sum-errors}
\end{equation}
\end{lemma}

\begin{proof}
Fix $0\le s<t<M$ and the affine coefficients of all stages before $t$. Since $X_s$ and $m_s$ depend only on earlier families, $\Delta_s$ is fixed. Therefore,
\[
\E_t[\Delta_s\Delta_t]=\Delta_s\E_t[\Delta_t]=0.
\]
The product $\Delta_s\Delta_t$ depends only on $\mathcal A_0,\ldots,\mathcal A_t$, so averaging over the earlier families using~\eqref{eq:successive-averaging} gives~\eqref{eq:orthogonality}. Similarly,
\[
\E[\Delta_t]=\E[\E_t[\Delta_t]]=0.
\]
Summing these mean-zero identities gives $\E[S]=0$. Expanding $S^2$ and using~\eqref{eq:orthogonality} gives the second identity in~\eqref{eq:sum-errors}.
\end{proof}

We now combine these moment bounds with Chebyshev's inequality to prove that the total candidate count is $\bigO(\kbar\log n)$ with probability $1-o(1)$.

\begin{theorem}
\label{thm:workload}
There are absolute constants $C_3,C_4>0$ such that, for every fixed typical graph and all sufficiently large $n$, we have
\begin{equation}
\Pp[W>C_3\kbar\log n]
\le C_4\left(
\frac1{\kbar\log n}+\frac1p+
\frac{M}{p\kbar^2\log^2 n}
\right)=o(1).
\label{eq:work-tail}
\end{equation}
\end{theorem}

\begin{proof}
Summing~\eqref{eq:recursion-centered} over $0\le t<M$ gives
\[
W-Z_0\le6A+\frac14(W-Z_M)+6S.
\]
Since $Z_M\ge0$, we obtain
\begin{equation}
\frac34W\le Z_0+6A+6S.
\label{eq:W-decomp}
\end{equation}
By~\eqref{eq:A-Z0}, there is an absolute constant $C_0>0$ such that $Z_0+6A\le C_0\kbar\log n$. Set $C_3:=2C_0$. By~\eqref{eq:W-decomp}, one has
\begin{equation}
\{W>C_3\kbar\log n\}
\subseteq\left\{S>\frac{C_0}{12}\kbar\log n\right\}.
\label{eq:event-inclusion}
\end{equation}

Using~\cref{lem:orthogonality} and averaging the bound in~\eqref{eq:centered-stage}, we obtain
\begin{align}
\E[S^2]
&=\sum_{t=0}^{M-1}\E[\Delta_t^2]\notag\\
&\le\sum_{t=0}^{M-1}\E[Z_t]
+\frac7p\sum_{t=0}^{M-1}\E[Z_t^2]\notag\\
&\le C\left(\kbar\log n+\frac{\kbar^2\log^2 n+M}{p}\right).
\label{eq:S-second-bound}
\end{align}
The last inequality follows from~\cref{lem:first-moment,lem:L2} and~\eqref{eq:A-Z0}, where $C$ is an absolute constant. Since $\E[S]=0$, Chebyshev's inequality and~\eqref{eq:event-inclusion} give
\begin{align*}
\Pp[W>C_3\kbar\log n]
&\le\Pp\!\left[|S|>\frac{C_0}{12}\kbar\log n\right]\\
&\le\frac{144\E[S^2]}{C_0^2\kbar^2\log^2 n}\\
&\le C_4\left(
\frac1{\kbar\log n}+\frac1p+
\frac{M}{p\kbar^2\log^2 n}
\right).
\end{align*}
Each term tends to zero because $p=\Theta(\sqrt{\kbar})$, $M=\bigO(\log n)$, and $\kbar\to\infty$.
\end{proof}

\subsection{Proof of the Main Theorem}
\label{sec:main}

Exact recovery requires that every true edge be retained and every non-edge be removed. We first show that the final rounds remove all non-edges with high probability. We then prove that every true edge is retained and combine the recovery and workload bounds.

\begin{lemma}
\label{lem:no-fp}
Let $\Ehat$ be the output of~\cref{alg:decoder} after $R=\log n$ final rounds. For every fixed typical graph and all sufficiently large $n$, we have
\begin{equation}
\Pp[\Ehat\setminus E\ne\varnothing]
\le n^{2-\log 24}=o(1).
\label{eq:false-positive-tail}
\end{equation}
\end{lemma}

\begin{proof}
Fix a non-edge $e=\{x,y\}$ and define
\[
\mathcal B_\tau:=\{e\in\PD_{L+\tau}\},
\qquad 0\le\tau\le R.
\]
For $1\le\tau\le R$, fix the affine coefficients of all stages before final round $\tau$. If $e\notin\PD_{L+\tau-1}$, it cannot appear in any later candidate set. Otherwise, $e$ is a non-defective pair of singleton blocks, and the round uses an independent affine family. By~\cref{lem:survival}, its retention probability is at most $1/24$. Since this round is global stage $L+\tau-1$, we have
\[
\E_{L+\tau-1}[\ind\{\mathcal B_\tau\}]
\le\frac1{24}\ind\{\mathcal B_{\tau-1}\}.
\]
Averaging over the earlier families using~\eqref{eq:successive-averaging} gives
\[
\Pp[\mathcal B_\tau]\le\frac1{24}\Pp[\mathcal B_{\tau-1}].
\]
Applying this inequality repeatedly and using $\Pp[\mathcal B_0]\le1$, we obtain
\[
\Pp[e\in\Ehat]\le24^{-R}=n^{-\log24},
\]
where $R=\log n$. A union bound over fewer than $n^2$ non-edges proves~\eqref{eq:false-positive-tail}. This bound tends to zero because $\log 24>2$.
\end{proof}

We now combine this result with edge preservation and the workload bound.

\begin{proof}[Proof of \cref{thm:main}]
Each of the $M=L+R$ affine families uses $p^2$ queries, all chosen before any outcomes are observed. Since $p^2=\Theta(\kbar)$ and $M=\bigO(\log n)$, the design uses $\bigO(\kbar\log n)$ non-adaptive queries. This proves~(i).

To prove edge preservation, fix any graph and any choice of affine coefficients, and let $\{x,y\}\in E$. For all sufficiently large $n$, we have $g_0\ge2$. Thus $\PD_0$ contains a pair of distinct blocks whose union contains both endpoints. If the endpoints lie in different blocks, take those two blocks. Otherwise, pair their common block with any other base block.

Every candidate pair whose union contains both endpoints is retained. If its slopes are equal, retention is automatic. Otherwise, its collision test contains the edge $\{x,y\}$ and therefore has positive outcome. At a non-final level, the retained pair generates all six pairs among its four child blocks. At least one generated pair has both endpoints in its union: if the endpoints lie in different child blocks, take those two blocks; otherwise, pair their common child block with another child.

By induction, a candidate pair containing both endpoints in its union exists at every level. At the singleton level, this pair is $\{x,y\}$ itself, and it is retained in every final round. Hence
\[
E\subseteq\Ehat
\]
for every graph and every choice of affine coefficients.

Recovery can therefore fail only if a non-edge survives. The bound in~\cref{lem:no-fp} holds uniformly over $G\in\cT(\epsilon_n)$. Combining it with~\cref{lem:structure} gives
\[
\Pp[\Ehat\ne E]
\le\Pp[G\notin\cT(\epsilon_n)]
+n^{2-\log 24}
=o(1).
\]
This proves~(ii).

Similarly, the uniform bound in~\cref{thm:workload} and~\cref{lem:structure} give
\begin{align*}
\Pp[W>C_3\kbar\log n]
&\le
\Pp[G\notin\cT(\epsilon_n)]
+\sup_{G\in\cT(\epsilon_n)}
\Pp[W>C_3\kbar\log n\mid G]
=o(1).
\end{align*}
The decoder examines $\sum_{t=0}^{M-1}Z_t\le W$ candidate pairs and attempts at most $6\sum_{t=0}^{L-1}Z_t\le6W$ child-pair insertions. Each examination uses a constant number of field operations and array accesses. Under the computational assumptions in~\cref{sec:problem}, initializing, iterating over, updating, and copying the candidate sets, together with producing the output, takes $\bigO(W+M)$ operations. The $M$ term accounts for stage overhead, and $W$ already includes the output size $Z_M$. Since $M=\bigO(\log n)$ and $\kbar\to\infty$, the decoding time is $\bigO(\kbar\log n)$ with probability $1-o(1)$. This proves~(iii).
\end{proof}

The query bound is optimal up to constant factors whenever $\log(n^2/\kbar)=\Omega(\log n)$. In this regime, the lower bound $\Omega(\kbar\log(n^2/\kbar))$ from~\cite{LFS19} becomes $\Omega(\kbar\log n)$ and matches our upper bound.

\section*{Declaration of AI Usage}
The author led the research direction and the development of the query design, decoder, and proof strategy. AI models (ChatGPT 5.5 Pro and ChatGPT-5.6 Sol Pro) assisted with exploring mathematical arguments, drafting intermediate proof sketches, and organizing and drafting parts of the exposition. All AI-generated material was carefully checked and substantially revised by the author, who take full responsibility for the manuscript and the validity of its results.

\section*{Acknowledgment}
The author thanks Prof. Jonathan Scarlett for helpful discussions and valuable suggestions on the presentation of this work.

\newpage

\bibliographystyle{alpha}
\bibliography{references}

\newpage

\appendix

\section{Proof of the Typical-Graph Lemma}
\label{app:typical}


Throughout this proof, the hierarchy is fixed and $G\sim\ER(n,q)$. We verify the four conditions in~\cref{def:typical}.

\begin{proof}[Proof of \cref{lem:structure}]
We use the binomial upper-tail bound
\begin{equation}
\Pp[X\ge t]\le\left(\frac{e\lambda}{t}\right)^t,
\qquad
X\sim\operatorname{Bin}(m,p'),\quad
\lambda=mp',\quad t>0,\quad t\ge\lambda.
\label{eq:bintail}
\end{equation}
The same bound holds with $\lambda$ replaced by any $\lambda'$ satisfying $\lambda\le\lambda'\le t$.

\paragraph{Edge count.}
Set $\epsilon_n:=\kbar^{-1/3}$. Since $|E|$ is binomial with mean $\kbar$, a two-sided Chernoff bound gives, for all sufficiently large $n$,
\[
\Pp\!\left[\bigl||E|-\kbar\bigr|>\epsilon_n\kbar\right]
\le2\exp\!\left(-\frac{\epsilon_n^2\kbar}{3}\right)
=2\exp\!\left(-\frac{\kbar^{1/3}}3\right)=o(1).
\]
Thus~\eqref{eq:struct-edges} holds with probability $1-o(1)$.

\paragraph{Internal edges at the base level.}
There are $M_0:=g_0\binom{n/g_0}{2}$ potential edges with both endpoints in the same base block. Hence $H_{g_0}\sim\operatorname{Bin}(M_0,q)$ and
\[
\E[H_{g_0}]
=\kbar\frac{n/g_0-1}{n-1}
\le\frac{\kbar}{g_0}\le\sqrt{\kbar}.
\]
Applying~\eqref{eq:bintail} with $t=3\sqrt{\kbar}$ gives
\[
\Pp[H_{g_0}>3\sqrt{\kbar}]
\le\left(\frac e3\right)^{3\sqrt{\kbar}}=o(1),
\]
which proves~\eqref{eq:struct-internal}. Since each level refines the previous one, $H_{g_j}\le H_{g_0}$ for every $j$. Thus, on the same event, every level has at most $3\sqrt{\kbar}$ internally defective blocks.

\paragraph{Uniform crossing degree.}
Fix a level $j$ and a block $B$ of size $n/g_j$. Let $X_B$ count the edges with exactly one endpoint in $B$. Then
\[
X_B\sim\operatorname{Bin}\!\left(
\frac n{g_j}\left(n-\frac n{g_j}\right),q
\right).
\]
For $n\ge3$, its mean satisfies
\[
\lambda_j:=\E[X_B]
\le\frac{qn^2}{g_j}
\le\frac{3\kbar}{g_j}
\le\frac{3\sqrt{\kbar}}{2^j}.
\]
Set $d_*:=24\sqrt{\kbar}$. By~\eqref{eq:bintail}, we have
\[
\Pp[X_B\ge d_*]
\le\left(\frac{e\lambda_j}{d_*}\right)^{d_*}
\le\left(\frac{e}{8\cdot2^j}\right)^{d_*}.
\]
Since $g_j<2^{j+1}\sqrt{\kbar}$, a union bound over all blocks and levels gives
\begin{align*}
\Pp\!\left[\max_{0\le j\le L}\Delta_{g_j}^{\times}\ge d_*\right]
&\le\sum_{j=0}^L g_j\left(\frac{e}{8\cdot2^j}\right)^{d_*}\\
&\le2\sqrt{\kbar}\left(\frac e8\right)^{d_*}
  \sum_{j=0}^{\infty}2^{-j(d_*-1)}\\
&\le4\sqrt{\kbar}\left(\frac e8\right)^{24\sqrt{\kbar}}
=o(1).
\end{align*}
The geometric sum is at most $2$ once $d_*\ge2$. This proves~\eqref{eq:struct-cross}.

\paragraph{Weighted internal edges over the hierarchy.} The idea of the proof is to choose a deterministic cutoff level so that, with high probability, no internal edges remain at this or any later level. We then apply Bernstein's inequality to the weighted sum over earlier levels, expressed as a sum of independent edge indicators with bounded weights. Define
\begin{equation}
\omega_n:=\sqrt{\log n},
\qquad
b_n:=\min\{\kbar\omega_n,n\},
\label{eq:cutoff-scale}
\end{equation}
and let $h$ be the smallest index in $\{0,\ldots,L\}$ such that $g_h\ge b_n$. These quantities depend only on $n$ and $\kbar$. For all sufficiently large $n$, both $n$ and $\kbar\omega_n$ exceed $g_0$, so $h\ge1$. By the choice of $h$,
\begin{equation}
g_h=2g_{h-1}<2b_n\le2\kbar\omega_n.
\label{eq:cutoff-upper}
\end{equation}
If $h=L$, then $H_{g_h}=0$ because the blocks are singletons. If $h<L$, then $b_n=\kbar\omega_n$ and
\[
\E[H_{g_h}]
=\kbar\frac{n/g_h-1}{n-1}
\le\frac{\kbar}{g_h}\le\frac1{\omega_n}.
\]
Markov's inequality therefore gives, in both cases,
\begin{equation}
\Pp[H_{g_h}>0]\le\frac1{\omega_n}=o(1).
\label{eq:cutoff-empty}
\end{equation}

For each potential edge $e=\{x,y\}\in\binom{[n]}2$, define
\begin{equation}
w_h(e):=\sum_{j=0}^{h-1}g_j
\ind\{x,y\text{ lie in the same level-}j\text{ block}\},
\qquad
\xi_e:=\ind\{e\in E\}.
\label{eq:cutoff-weight}
\end{equation}
Changing the order of summation gives
\begin{equation}
K_h:=\sum_{e\in\binom{[n]}2}\xi_e w_h(e)
=\sum_{j=0}^{h-1}g_jH_{g_j}.
\label{eq:cutoff-sum}
\end{equation}
The weights are deterministic and satisfy
\begin{equation}
0\le w_h(e)\le\sum_{j=0}^{h-1}g_j=g_h-g_0<g_h.
\label{eq:cutoff-weight-bound}
\end{equation}

To compute the average weight, let $U$ be a uniformly chosen element of $\binom{[n]}2$, and write $\Pp_U$ and $\E_U[\cdot]$ for probability and expectation over $U$ alone. At a level with $g$ blocks,
\begin{equation}
\Pp_U[U\text{ is internal at this level}]
=\frac{g\binom{n/g}{2}}{\binom n2}
=\frac{n/g-1}{n-1}\le\frac1g.
\label{eq:internal-prob}
\end{equation}
It follows that $\E_U[w_h(U)]\le h\le\log n$. Hence
\begin{equation}
\mu_h:=\E[K_h]=q\sum_e w_h(e)
=\kbar\,\E_U[w_h(U)]\le\kbar\log n.
\label{eq:cutoff-mean}
\end{equation}

Since the cutoff and the weights are deterministic, the variables $\xi_e w_h(e)$ are independent. Their total variance satisfies
\begin{align}
V_h&:=\sum_e\Var(\xi_e w_h(e))
\le q\sum_e w_h(e)^2
\le qg_h\sum_e w_h(e)
=g_h\mu_h
\le\kbar g_h\log n.
\label{eq:cutoff-var}
\end{align}
Each summand differs from its mean by at most $g_h$. Bernstein's inequality therefore gives, for $t>0$,
\begin{equation}
\Pp[K_h-\mu_h\ge t]
\le\exp\!\left(-\frac{t^2}{2V_h+\frac23g_ht}\right).
\label{eq:cutoff-bernstein-general}
\end{equation}
Taking $t=2\kbar\log n$ and using~\eqref{eq:cutoff-mean}, \eqref{eq:cutoff-var}, and~\eqref{eq:cutoff-upper}, we obtain
\begin{align}
\Pp[K_h>3\kbar\log n]
\le\exp\!\left(-c\frac{\kbar\log n}{g_h}\right)
\le\exp\!\left(-c'\frac{\log n}{\omega_n}\right)
=\exp(-c'\sqrt{\log n})=o(1)
\label{eq:cutoff-bernstein}
\end{align}
for absolute constants $c,c'>0$.

The Bernstein bound above is unconditional. On the event $H_{g_h}=0$, every finer level also has no internal edges. Thus
\[
\sum_{j=0}^{L-1}g_jH_{g_j}=K_h
\]
on this event. Combining~\eqref{eq:cutoff-empty} and~\eqref{eq:cutoff-bernstein} gives
\begin{align*}
\Pp\!\left[\sum_{j=0}^{L-1}g_jH_{g_j}>3\kbar\log n\right]
\le\Pp[H_{g_h}>0]+\Pp[K_h>3\kbar\log n]
\le\frac1{\sqrt{\log n}}+\exp(-c'\sqrt{\log n})
=o(1).
\end{align*}
This proves~\eqref{eq:struct-weighted}.

\paragraph{Combining the estimates.}
A union bound over the four failure events shows that $G\in\cT(\epsilon_n)$ with probability $1-o(1)$. On this event, $|E|\le2\kbar$ for all sufficiently large $n$. By~\cref{lem:Dg-det},
\[
\sum_{j=0}^{L-1}D_{g_j}
\le L|E|+\sum_{j=0}^{L-1}g_jH_{g_j}
\le5\kbar\log n.
\]
This also proves~\eqref{eq:sum-Dg} and completes the proof.
\end{proof}

\section{Deferred Proofs}
\label{app:covariance}


This appendix proves the two propositions used in the covariance bound. Throughout, we fix a stage, its partition into $g$ blocks, and a graph $G\in\cT(\epsilon_n)$. All probabilities refer to the affine family of that stage.

For distinct non-defective pairs $e,f$, define
\begin{equation}
D:=e\setminus f,
\qquad
R_0:=[g]\setminus(e\cup f).
\label{eq:block-decomp}
\end{equation}
The sets $f,D,R_0$ form a disjoint partition of $[g]$. We call the blocks in $e\cup f$ \emph{special}. If $e$ and $f$ are disjoint, then $|D|=2$ and $|e\triangle f|=4$. If they share a block, then $|D|=1$ and $|e\triangle f|=2$. On $\cC_e\cap\cC_f$, a shared block belongs to both collision tests. This notation allows us to treat both cases in the same proof.

\subsection{\texorpdfstring{Proof of \cref{prop:sep-prob}}{Proof of the separation bound}}
\label{app:sep-prob}

We express the failure of separation as a union of events, each with probability at most $1/p$, and then apply a union bound.

\begin{proof}
By the definition of $\Omega_{ef}$,
\begin{align}
\Omega_{ef}^c\subseteq{}&
\cC_e^c\cup\cC_f^c
\cup\bigl(\cC_e\cap\cC_f\cap\{r_e^*=r_f^*\}\bigr)
{}\cup\bigcup_{w\in D}\{w\in T_f\}
\cup\bigcup_{w\in f\setminus e}\{w\in T_e\}.
\label{eq:covering}
\end{align}
Indeed, if both pairs have distinct slopes, separation fails only if the collision points have the same first coordinate or one of the required block exclusions fails.

Each equal-slope event has probability $1/p$. For $w\in D$, we have $w\notin f$, so~\cref{prop:composition} gives
\[
\Pp[w\in T_f\mid\cC_f,\ T_f=(r',s')]=\frac1p
\]
for every possible collision point. By the convention $\{w\in\dagger\}=\varnothing$, the event $\{w\in T_f\}$ implies $\cC_f$. The law of total probability therefore gives
\begin{align*}
\Pp[w\in T_f]
&=\sum_{(r',s')\in\F_p^2}
\Pp[\cC_f,\ T_f=(r',s')]
\Pp[w\in T_f\mid\cC_f,\ T_f=(r',s')]
=\frac1p\Pp[\cC_f]\le\frac1p.
\end{align*}
Terms with zero-probability conditioning events are omitted. Exchanging $e$ and $f$ gives the same bound for each $w\in f\setminus e$.

For disjoint pairs,~\cref{prop:first_cordinate}(ii) gives
\[
\Pp[r_e^*=r_f^*\mid\cC_e\cap\cC_f]=\frac1p.
\]
If the pairs share a block $u$,~\cref{prop:first_cordinate}(iii) gives the same probability when $(a_u,b_u)$ is also fixed. Averaging over these coefficients conditional on $\cC_e\cap\cC_f$ gives the same equality. Thus, in both cases,
\[
\Pp[\cC_e\cap\cC_f\cap\{r_e^*=r_f^*\}]
=\frac1p\Pp[\cC_e\cap\cC_f]\le\frac1p.
\]
There are $3+|D|+|f\setminus e|=3+|e\triangle f|$ events on the right-hand side of~\eqref{eq:covering}. A union bound now gives
\[
\Pp[\Omega_{ef}^c]\le\frac{3+|e\triangle f|}{p}\le\frac7p.
\qedhere
\]
\end{proof}

\subsection{\texorpdfstring{Proof of \cref{prop:product}}{Proof of the product bound}}
\label{app:product}

On $\Omega_{ef}$, we first fix the special coefficients and the full set of blocks included in the first collision test. This determines the first outcome. Two-point uniformity then describes the remaining randomness in the second test, and monotonicity bounds its probability of a positive outcome.

\begin{proof}
For $Q\subseteq[g]\setminus f$, let 
\[
\Phi_f(Q):=Y\!\left(\bigcup_{x\in f\cup Q}\cG_x\right)
\]
be the outcome of the query formed by the blocks in $f\cup Q$. Since $G$ is fixed, this outcome is deterministic. Adding blocks cannot remove an edge, so
\begin{equation}
Q\subseteq Q'\quad\Longrightarrow\quad
\Phi_f(Q)\le\Phi_f(Q').
\label{eq:Phi-monotone}
\end{equation}
For a finite set $X$ and $B\subseteq X$, define
\[
\mu_X(B):=\left(\frac1p\right)^{|B|}
\left(1-\frac1p\right)^{|X|-|B|}.
\]
This is the probability of selecting exactly $B$ when each element of $X$ is included independently with probability $1/p$. In particular, $\sum_{B\subseteq X}\mu_X(B)=1$. By~\cref{prop:composition} and the disjoint decomposition $[g]\setminus f=D\cup R_0$,
\begin{equation}
\rho_f=\sum_{B\subseteq R_0}\mu_{R_0}(B)
\sum_{J\subseteq D}\mu_D(J)\Phi_f(B\cup J).
\label{eq:rho-expanded}
\end{equation}

\paragraph{Partition of the separation event.}
Let $\mathbf C$ be the vector of coefficient pairs of the special blocks. Every condition defining $\Omega_{ef}$ depends only on $\mathbf C$. Let $\mathcal S$ be the set of possible values of $\mathbf C$ that satisfy these conditions. Then
\[
\Omega_{ef}=\bigcup_{\mathbf c\in\mathcal S}\{\mathbf C=\mathbf c\},
\]
where the events in the union are disjoint. Fix $\mathbf c\in\mathcal S$ and write the corresponding collision points as
\[
T_e(\mathbf c)=(r_{\mathbf c},s_{\mathbf c}),
\qquad
T_f(\mathbf c)=(r'_{\mathbf c},s'_{\mathbf c}).
\]
By the definition of $\mathcal S$,
\begin{equation}
r_{\mathbf c}\ne r'_{\mathbf c},
\qquad
w\notin T_f(\mathbf c)\ (w\in D),
\qquad
w\notin T_e(\mathbf c)\ (w\in f\setminus e).
\label{eq:special-separation}
\end{equation}
For $w\in R_0$, define the membership indicators
\[
U_w^{\mathbf c}:=\ind\{P_w(r_{\mathbf c})=s_{\mathbf c}\},
\qquad
V_w^{\mathbf c}:=\ind\{P_w(r'_{\mathbf c})=s'_{\mathbf c}\}.
\]
For each $\mathbf z=(z_w)_{w\in R_0}\in\{0,1\}^{R_0}$, define
\[
A_{\mathbf c,\mathbf z}:=
\{\mathbf C=\mathbf c\}\cap
\bigcap_{w\in R_0}\{U_w^{\mathbf c}=z_w\}.
\]
As $\mathbf c$ and $\mathbf z$ vary, these events form a disjoint partition of $\Omega_{ef}$. Each event has positive probability: after the special coefficients are fixed, the remaining coefficients are still independent and uniform, and every prescribed membership pattern has positive probability.

On $A_{\mathbf c,\mathbf z}$, the blocks included in the first test are exactly
\[
e\cup\{w\in R_0:z_w=1\}.
\]
The blocks in $e$ belong to their collision test, the blocks in $f\setminus e$ are excluded by~\eqref{eq:special-separation}, and $\mathbf z$ determines all remaining memberships. Since $G$ is fixed and $\cC_e$ holds, both $Y_{T_e}$ and $I_e$ are fixed on this event.

\paragraph{Conditional distribution of the second test.}
The coefficient pairs of the blocks in $R_0$ are independent of one another and of $\mathbf C$. For fixed $\mathbf c$, the variables $U_w^{\mathbf c}$ and $V_w^{\mathbf c}$ depend only on the coefficients of block $w$. Therefore, for every $\mathbf y\in\{0,1\}^{R_0}$,
\begin{align}
&\Pp[(V_w^{\mathbf c})_{w\in R_0}=\mathbf y\mid A_{\mathbf c,\mathbf z}]
=
\prod_{w\in R_0}
\frac{\Pp[U_w^{\mathbf c}=z_w,\ V_w^{\mathbf c}=y_w]}
{\Pp[U_w^{\mathbf c}=z_w]}
=\prod_{w\in R_0}
\Pp[V_w^{\mathbf c}=y_w\mid U_w^{\mathbf c}=z_w].
\label{eq:cell-factorization}
\end{align}
All denominators are positive. Since $r_{\mathbf c}\ne r'_{\mathbf c}$,~\cref{prop:two-point} implies that $U_w^{\mathbf c}$ and $V_w^{\mathbf c}$ are independent Bernoulli variables with parameter $1/p$. Hence~\eqref{eq:cell-factorization} gives
\begin{equation}
\Pp[(V_w^{\mathbf c})_{w\in R_0}=\mathbf y\mid A_{\mathbf c,\mathbf z}]
=\mu_{R_0}(\{w:y_w=1\}).
\label{eq:cell-law}
\end{equation}
Thus, after fixing the special coefficients and all memberships in the first test, the blocks in $R_0$ are still included independently in the second test with probability $1/p$.

\paragraph{Bounding the second retention probability.}
On $A_{\mathbf c,\mathbf z}$, both collision events hold. The second test contains exactly the blocks in $f\cup B$, where
\[
B:=\{w\in R_0:V_w^{\mathbf c}=1\}.
\]
The blocks in $D$ are excluded by~\eqref{eq:special-separation}. Since $\cC_f$ holds, the pair $f$ is retained exactly when this test is positive. By~\eqref{eq:cell-law},
\[
\Pp[\cE_f\mid A_{\mathbf c,\mathbf z}]
=\sum_{B\subseteq R_0}\mu_{R_0}(B)\Phi_f(B).
\]
For every $B\subseteq R_0$ and $J\subseteq D$, monotonicity gives $\Phi_f(B)\le\Phi_f(B\cup J)$. Averaging first over $J$ and then over $B$, we obtain
\begin{align}
\Pp[\cE_f\mid A_{\mathbf c,\mathbf z}]
&\le\sum_{B\subseteq R_0}\mu_{R_0}(B)
\sum_{J\subseteq D}\mu_D(J)\Phi_f(B\cup J)
=\rho_f.
\label{eq:cell-domination}
\end{align}

Let $\Gamma$ be the set of indices $(\mathbf c,\mathbf z)$ for which $A_{\mathbf c,\mathbf z}\subseteq\cE_e$. Since $I_e$ is constant on each event in the partition,
\[
\cE_e\cap\Omega_{ef}
=\bigcup_{(\mathbf c,\mathbf z)\in\Gamma}A_{\mathbf c,\mathbf z}
\]
is a disjoint union. Using~\eqref{eq:cell-domination}, the law of total probability, and~\eqref{eq:rho-le-q}, we get
\begin{align*}
\Pp[\cE_e\cap\cE_f\cap\Omega_{ef}]
=\sum_{(\mathbf c,\mathbf z)\in\Gamma}
\Pp[A_{\mathbf c,\mathbf z}]
\Pp[\cE_f\mid A_{\mathbf c,\mathbf z}]
\le\rho_f\sum_{(\mathbf c,\mathbf z)\in\Gamma}
\Pp[A_{\mathbf c,\mathbf z}]
=\rho_f\,\Pp[\cE_e\cap\Omega_{ef}] \le q_eq_f.
\end{align*}
The argument applies to both disjoint pairs and pairs sharing a block, proving~\eqref{eq:product-bound}.
\end{proof}




\end{document}